\documentclass[preprint,number,sort&compress,12pt]{elsarticle}

\makeatletter
\def\@author#1{\g@addto@macro\elsauthors{\normalsize%
    \def\baselinestretch{1}%
    \upshape\authorsep#1\unskip\textsuperscript{%
      \ifx\@fnmark\@empty\else\unskip\sep\@fnmark\let\sep=,\fi
      \ifx\@corref\@empty\else\unskip\sep\@corref\let\sep=,\fi
      }%
    \def\authorsep{\unskip,\space}%
    \global\let\@fnmark\@empty
    \global\let\@corref\@empty  
    \global\let\sep\@empty}%
    \@eadauthor={#1}
}
\makeatother

\usepackage{amssymb}
\usepackage{amsfonts}
\usepackage{amsmath}
\usepackage{amsthm}
\usepackage{epsfig}
\usepackage[latin1]{inputenc}

\parskip=\medskipamount
\newtheorem{theorem}{Theorem}[section]
\newtheorem{lemma}[theorem]{Lemma}
\newtheorem{corollary}[theorem]{Corollary}
\newtheorem{proposition}[theorem]{Proposition}
\theoremstyle{definition}
\newtheorem{definition}[theorem]{Definition}
\theoremstyle{remark}
\newtheorem{remark}[theorem]{Remark}
\newtheorem{example}[theorem]{Example}

\newcommand{\C}{\mathbb{C}}

\newcommand{\N}{\mathbb{N}}
\newcommand{\R}{\mathbb{R}}
\newcommand{\Z}{\mathbb{Z}}
\newcommand{\zak}{\mathrm{Z}}

\newcommand{\gabor}{\mathcal G(g,\alpha,\beta)}

\newcommand{\esssup}{\mathop{{\rm ess}\,{\rm sup}}\limits}
\newcommand{\essinf}{\mathop{{\rm ess}\,{\rm inf}}\limits}
\newcommand{\abs}[1]{\left|#1\right|}
\newcommand{\norm}[1]{\left\|#1\right\|}
\newcommand{\zweinorm}[1]{\left\|#1\right\|_2}

\newcommand{\id}{\mathrm{id}}

\newcommand{\ds}{\displaystyle}

\begin{document}

\begin{frontmatter}
\title
{Zak transforms and Gabor frames 
 of totally positive functions and exponential B-splines\tnoteref{t1}}

\tnotetext[t1]{AMS subject classification 2000: 42C15,41A15,42C40}

\author{Tobias Kloos\corref{cor1}}
\ead{tobias.kloos@math.tu-dortmund.de}
\address{Faculty of Mathematics, TU Dortmund, D-44227 Dortmund}
\cortext[cor1]{Corresponding author}

\author{Joachim St\"ockler}
\ead{stoeckler@math.tu-dortmund.de}
\address{Faculty of Mathematics, TU Dortmund, D-44227 Dortmund}

\begin{abstract}
We study totally positive (TP) functions of finite type 
and exponential B-splines 
as window functions for
Gabor frames. We establish the connection
of the Zak transform of these two classes of 
functions and prove that the Zak transforms have only one 
zero in their
fundamental domain of quasi-periodicity.
Our proof is based on the
variation-diminishing property of shifts of exponential B-splines.
For the exponential B-spline
$B_m$ of order $m$, we determine a large set of 
lattice parameters $\alpha,\beta>0$ such that 
the Gabor family 
$\mathcal{G}(B_m,\alpha,\beta)$ of time-frequency shifts
$e^{2\pi i l\beta} B_m(\cdot-k\alpha)$,
$k,l\in\Z$, is a frame for $L^2(\R)$. 
By the connection of its Zak transform
to the Zak transform of TP functions of finite type, our  result 
provides an alternative proof 
that
TP functions of finite type provide Gabor frames for all lattice parameters
with $\alpha\beta<1$. For even two-sided exponentials 
$g(x)=\tfrac{\lambda}{2}e^{-\lambda|x|}$ and the related exponential B-spline
of order $2$, we find lower frame-bounds $A$, which show the
asymptotically linear decay  $A\sim (1-\alpha\beta)$ as the density 
$\alpha \beta$ of the time-frequency 
lattice  tends to the critical density $\alpha\beta=1$. 
\end{abstract}


\end{frontmatter}


\section*{Introduction}

The Gabor transform provides an important tool for the ana\-lysis of 
a given signal $f:\R\to\C$ in time and frequency. It is a fundamental tool 
for time-frequency analysis and serves various purposes, such as signal denoising or
compression. 
A window function $g\in L^2(\R)$ has
 time-frequency shifts
$$
    M_\xi T_y g(x)= e^{2\pi i \xi x} g(x-y),\qquad \xi,y\in\R.
$$
The Gabor transform of a square-integrable signal $f$ is defined as 
$$
   \mathcal{S}_g f(k\alpha,l\beta)
	= \langle f,M_{l\beta} T_{k\alpha} g\rangle,\qquad k,l\in\Z,
$$
where the parameters 
$(k\alpha,l\beta)\in\alpha\Z\times \beta\Z$ of the time-frequency shifts of
$g$ form a lattice in $\R^2$, with 
lattice parameters $\alpha,\beta>0$. The family
$$\mathcal G(g,\alpha,\beta) := \{M_{l\beta} T_{k\alpha}\,g \mid k,l\in\Z \}$$
is called a Gabor family.
 An important problem in Gabor analysis is to 
determine, for a given window function $g\in L^2(\R)$ and lattice 
parameters $\alpha,\beta>0$, if the Gabor family $\gabor$ is a frame for $L^2(\R)$.
This means that there
exist constants $A,B>0$, which depend on $g,\alpha,\beta$, such that for 
every $f\in L^2(\R)$ we have
\begin{equation}\label{eq:framecond}
   A\|f\|^2 \le \sum_{k,l\in\Z} |\langle f, M_{l\beta} T_{k\alpha} g\rangle|^2
	\le B\|f\|^2.
\end{equation}
Moreover, it is interesting for the selection of sampling and modulation rates
in signal analysis to
 know the whole \textit{frameset} 
$$\mathcal F_g := \{ (\alpha,\beta) \in \R^2_+ 
\mid \mathcal G(g,\alpha,\beta) \textrm{ is a frame} \}$$
of a given function $g\in L^2(\R)$.
Until 2012, the framesets of only a few functions, like the one- and two-sided exponential, the Gaussian and the hyperbolic secant, 
were known. Then it was proved in \cite[Theorem 1]{GroeStoe:2012} 
that  the frameset of a large class of functions,
the \textit{totally positive functions of finite type} 
$m\geq 2$, is given by
$$\mathcal F_g = \mathbb H := \{(\alpha,\beta) \in \R^2_+ \mid0<\alpha\beta<1\}.
$$
This result complements the known fact, that for continuous window functions of the \textit{Wiener space}, 
\begin{equation*}
W(\R):= \{ g\in L^{\infty}(\R) \mid \norm{g}_W=\sum_{n\in\Z} \esssup_{x\in[0,1]} \abs{g(x+n)} < \infty \}, 	
\end{equation*}
the frameset is a subset of $\mathbb H$.


\vspace{0.2cm}
In order to describe Gabor families of a window function $g$, 
a very helpful tool is the Zak transform
\begin{equation*}
\mathrm Z_{\alpha}g(x,\omega) := \sum_{k\in\Z} g(x-k\alpha) e^{2\pi ik\alpha\omega},
\qquad (x,\omega)\in\R^2.
\end{equation*}
In Approximation Theory, the Zak transform  was 
used by Schoenberg \cite{Schoenberg:1973} 
in connection with cardinal spline interpolation.
For a polynomial 
B-spline $B_m$ of degree $m-1$, Schoenberg called $Z_1B$ 
the {\em exponential Euler spline}. 
The Zak transform $Z_\alpha g$ has the properties
\begin{equation}\label{eq:Zak1}
\mathrm Z_{\alpha}g(x,\omega+\tfrac{1}{\alpha}) = \mathrm Z_{\alpha}g(x,\omega),
\qquad
\mathrm Z_{\alpha}g(x+\alpha,\omega) = 
e^{2\pi i\alpha \omega}\, \mathrm Z_{\alpha}g(x,\omega).
\end{equation}
Therefore, its values in the lattice cell $[0,\alpha)\times [0,\tfrac{1}{\alpha})$
define $Z_\alpha g$ completely.  A well-known result for Gabor families $\gabor$, with
$\alpha=1$, $\beta=1/N$, and $N\in\N$, states that
the values  
$$
   A_{\rm opt}=\essinf_{x,\omega\in[0,1)} \sum_{j=0}^{N-1} 
	|Z_1 g(x,\omega+\tfrac{j}{N})|^2,\qquad
	   B_{\rm opt}=\esssup_{x,\omega\in[0,1)} \sum_{j=0}^{N-1} 
	|Z_1 g(x,\omega+\tfrac{j}{N})|^2,
$$
are the optimal constants $A,B$ in the inequality \eqref{eq:framecond}, see
 \cite{Dau:1990}, \cite{Janssen:2003.2}. For rational values of 
$\alpha\beta$, a connection of the Zak transform $Z_\alpha g$ with 
the frame-bounds of the Gabor family $\gabor$ 
was given by Zibulsky and Zeevi \cite{ZibZee:1997}. 
Therefore, the presence and the
location of zeros
of $Z_\alpha  g$ is relevant for the existence of lower frame-bounds in 
\eqref{eq:framecond}. For example, the celebrated Balian-Low theorem
\cite{Dau:1990} 
states that a Gabor family at the critical density $\alpha\beta=1$ 
cannot be a frame, if the window function $g$ or its Fourier transform 
$$
    \hat g(\omega)= \int_\R g(x) e^{-2\pi i x \omega}\,dx
$$
is continuous and in $W(\R)$. The proof in \cite{Janssen:1982} uses the
topological argument, that every continuous 
function with the property \eqref{eq:Zak1} must have a zero in every lattice cell.

Motivated by the results in \cite{GroeStoe:2012}, we study three problems in Gabor analysis.
First, in Section~3, we investigate zero properties
of the Zak transform $Z_\alpha g$, where $g$ is a TP function of finite type.
Our main tools are methods from the theory
of Tschebycheffian splines which is elaborated in \cite{Schum:1981} and 
\cite{Ron:1987}. We 
establish the connection of $Z_\alpha g$ with the Zak transform $Z_1 B_m$
of an exponential B-spline in Theorem \ref{b-splinezak}. For exponential B-splines $B_m$
of order $m\ge 2$
(and the slightly larger class of periodic exponential B-splines defined in 
\cite{Ron:1987}), we show in Theorem \ref{zakzeroesB} 
that the Zak transform $Z_1 B_m$ has exactly one 
zero  $(x,\omega)\in[0,1)^2$, and this zero is located somewhere on 
the line $\omega=\tfrac{1}{2}$. The main argument for the proof is the 
variation-dimishing property of Tschebycheffian B-splines. As a corollary, 
we obtain that $Z_\alpha g$ has only one 
zero  $(x,\omega)\in[0,\alpha)\times [0,\tfrac{1}{\alpha})$, and this zero is 
located  on 
the line $\omega=\tfrac{\alpha}{2}$. These results add two new families of examples 
to the study of Zak transforms with few zeros by Janssen \cite{Janssen:2003.2}.

In Section 4, we prove that the Gabor families $\mathcal{G}(B_m,\alpha,\beta)$
are  frames for $L^2(\R)$ for certain values $\alpha,\beta>0$.
Another important problem in Gabor analysis is the 
asymptotic behaviour of the 
lower frame-bound $A$ in \eqref{eq:framecond} 
 near the critical density $\alpha\beta\approx 1$. This information is required
in practical situations where only a minimal rate of oversampling is allowed.
It is known that for every continuous window function 
in $W(\R)$, the lower frame-bound tends to $0$ as $\alpha\beta$ tends to $1$.
The only window functions (and scaled versions thereof) in the literature, 
where the asymptotic behaviour of the 
lower frame-bound was specified near the critical density,
 are the Gaussian $g(x)=e^{-\pi x^2}$ and the hyperbolic secant 
$g(x)=(\cosh \pi x)^{-1}$. It was proved in
\cite{BoGroLy:2010} that constants $c_1,c_2>0$ exist  
such that the optimal lower frame-bound $A_{\rm opt}$ 
of $\mathcal{G}(g,\alpha,\beta)$
satisfies 
$$
   c_1 (1-\alpha\beta) \le A_{\rm opt} \le c_2(1-\alpha\beta)\quad \hbox{for}\quad
	 \tfrac{1}{2}<\alpha\beta < 1.
$$
In Section 5, we consider
all symmetric exponential B-splines of order 2 and 
give explicit lower frame-bounds that exhibit the same linear
asymptotic decay when $\alpha=1$ and $\beta$ tends to $1$.
The main ingredient in the proofs of the results in sections~4 and 5 is the fact that 
collocation matrices of Tschebycheffian B-splines are almost strict totally positive
(ASTP) matrices, in the terminology of Gasca et al. \cite{GasMicPen:1992}.
This property is equivalent to the Schoenberg-Whitney conditions
 \cite{Schum:1981}. 

The connection of the Zak transforms of a TP function $g$ of finite type $m$ 
and an exponential B-spline $B_m$ of order $m$ in Theorem \ref{b-splinezak}
opens a new corridor for further study of Gabor frames $\gabor$.
In combination with the results for $B_m$ in section~4, we obtain an alternative proof
for \cite[Theorem 1]{GroeStoe:2012} that TP functions of finite type provide Gabor frames for all lattice parameters
with $\alpha\beta<1$. Moreover it allows us to give lower frame-bounds for
$\gabor$ as well, where $g$ is the two-sided exponential $g(x)=\tfrac{\lambda}{2}e^{-\lambda|x|}$
which is TP of finite type $2$.

In the first two sections, we give a brief introduction and provide enough 
background on total positivity and Tschebycheffian B-splines for our purpose. 
For non-specialists in Approximation Theory, this may be helpful for reading the remaining sections.
The results on Gabor frames which are needed in this article are contained in sections 3-5.
For a detailed exposition on Gabor frames, we recommend the 
monographs \cite{Groech:2001}, \cite{Chris:2003}, a very short account 
is contained in the introduction in \cite{GroeStoe:2012}.


\section{Totally positive functions, matrices and sequences}

In this section, we give a short review on total positivity. We refer to \cite{Karl:1968}
for a comprehensive exposition and include some more recent developments which will be helpful in 
our context.

Totally positive functions were 
 introduced by Schoenberg in 1947 \cite{Schoenberg:1947}. Schoenberg and Whitney 
\cite{SchoenWhit:1953} laid an important foundation 
for their applications in Approximation Theory, e.g. in spline interpolation. 

\begin{definition}[Totally positive (TP) function, \cite{Schoenberg:1947}]
A measurable, non-constant function $g:\R\rightarrow\R$ is called totally positive (TP), if for every $N\in\N$ and two sets of real numbers
$$x_1<x_2<\ldots<x_N\ ,\ \ \ \ \ y_1<y_2<\ldots<y_N,$$
the corresponding matrix $A=\bigl(g(x_j-y_k)\bigr)_{j,k=1}^N$ has a non-negative determinant.\\
If the determinant is always strictly positive, $g$ is called strictly totally positive (STP).
\end{definition}

An example of a TP function is the two-sided exponential function, given by $g(x) = e^{-\abs{x}}$. The Gaussian
$g(x)=e^{-\pi x^2}$ is an STP function.
Schoenberg characterized TP functions by their Laplace transform \cite{Schoenberg:1947}
and gave a characterization of all integrable TP functions \cite{Schoenberg:1951}.

\begin{theorem}[\cite{Schoenberg:1947}, \cite{Schoenberg:1951}]
A function $g:\R\rightarrow\R$, which is not an exponential $g(x) = Ce^{ax}$ with $C,a\in\R$, is a TP function, if and only if 
its two-sided Laplace transform exists in a strip $S=\{s\in\C\mid\alpha < \mathrm{Re}\,s < \beta\}$ with $-\infty\leq\alpha<\beta\leq\infty$ 
and is given by
$$(\mathcal Lg)(s) = \int_{-\infty}^{\infty} g(t) e^{-st}\,dt = Cs^{-n}e^{\gamma s^2-\delta s} 
\prod_{\nu=1}^{\infty} \frac{e^{a_{\nu}^{-1}s}}{1+a_{\nu}^{-1}s},$$
where $n\in\N_0$ and $C,\gamma,\delta,a_{\nu}$ are real parameters with
$$C>0,\ \ \gamma\geq0,\ \ a_{\nu}\neq 0,\ \ 0<\gamma+\sum_{\nu=1}^{\infty}\left(\tfrac{1}{a_{\nu}}\right)^2<\infty.$$
Moreover, $g$ is integrable and TP, if and only if its Fourier transform is given by
$$\hat{g}(\omega) =  \int_{-\infty}^{\infty} g(t) e^{-2\pi it\omega}\,dt 
=C e^{-\gamma\omega^2} e^{-2\pi i\delta\omega} 
\prod_{\nu=1}^{\infty}\frac{e^{2\pi ia_{\nu}^{-1}\omega}}{1+2\pi ia_{\nu}^{-1}\omega}$$
with the same conditions on $C,\gamma,\delta,a_{\nu}$ as above.
\end{theorem}

The subset of  TP functions of finite type $m\in\N$ is given by their Fourier transform 
\begin{equation}\label{eq:TPfinite}
\hat{g}(\omega) = C\,\prod_{\nu=1}^m (1+2\pi i\omega a_{\nu}^{-1})^{-1},\ \ a_{\nu}\in\R\setminus\{0\},\ \ C>0.
\end{equation}
In \cite{GroeStoe:2012}, these are considered being window functions of Gabor frames $\gabor$.
For example, the one- and two-sided exponentials 
$$ g_1(x)=e^{-x}\chi_{[0,\infty)}(x),\qquad g_2(x)=e^{-|x|}$$ 
are TP functions of finite type $1$ and $2$, respectively.
The main result in \cite{GroeStoe:2012} states that $\gabor$ is a frame for all
lattice parameters $\alpha,\beta>0$ with $\alpha\beta<1$. 
Its proof is based on the following property of TP functions of finite type, 
which was characterized by Schoenberg and Whitney.

\begin{theorem}[\cite{SchoenWhit:1953}]\label{SchoenWhit}
Let $g\in L^1(\R)$ be a TP function of finite type $m=m_1+m_2\in\N$, where $m_1\in\N_0$ is 
the number of positive and $m_2\in\N_0$ is the number of negative $a_{\nu}$ in \eqref{eq:TPfinite}.
Then for every $N\in\N$ and two sets of real numbers
$$x_1<x_2<\ldots<x_N\ ,\ \ \ \ \ y_1<y_2<\ldots<y_N,$$
the determinant of the matrix $A=\bigl(g(x_j-y_k)\bigr)_{j,k=1}^N$ is strictly positive, 
 if and only if 
\begin{align}\label{S-W-prop}
x_{j-m_1}<y_j<x_{j+m_2}
\end{align} 
for all $j=1,\ldots,N$. Here, we let $x_k = -\infty$, if $k<1$, and $x_k = \infty$, if $k>N$.
\end{theorem}

The property (\ref{S-W-prop}) is called Schoenberg-Whitney condition, or in terms of interpolation, interlacing property.

The notion of total 
positivity is also used in matrix linear algebra, see \cite{Karl:1968} and \cite{Pinkus:2010}.
A rectangular matrix $A=(a_{j,k})\in\R^{m\times n}$ is called totally positive (TP), 
if all its minors are non-negative, that is
\begin{align*}
 &\mathrm{det}\bigl((a_{i_k,j_l})_{k,l=1,\ldots,s}\bigr) \geq 0\qquad\textrm{for all}\\
&1\leq s\leq \min\{m,n\},\quad 1\leq i_1<\ldots<i_s\leq m,\quad 1\leq j_1<\ldots<j_s\leq n.
\end{align*}
It is called strictly totally positive (STP), if all its minors are strictly positive.
Examples of TP matrices are the collocation matrices of polynomial B-splines or Tschebycheffian B-splines,
see \cite[Theorem 9.34]{Schum:1981} or Theorem~\ref{det>0} below. 
The following generalization of STP matrices is useful in our context.

\begin{definition}[Almost strict totally positive (ASTP) matrix, \cite{GasMicPen:1992}]
An $m\times n$ TP matrix $A$ is said to be almost strict totally positive (ASTP), if it satisfies the following two conditions:
\begin{itemize}
 \item[(1) ]Any minor of $A$ with consecutive rows and columns of $A$ 
is positive, if and only if the diagonal entries of the minor are positive.
 \item[(2) ]If $A$ has a zero row or column, then the subsequent rows or columns of $A$ are also zero.
\end{itemize}
\end{definition}

It was proved in \cite{GasMicPen:1992} that an arbitrary 
 minor of an ASTP matrix $A$ is positive, if and only if the diagonal entries
of the corresponding submatrix are positive. 

Yet another notion of total positivity is used for sequences of functions \newline
$(u_1,\ldots,u_n)$
defined on a  set
$K\subset\R$. For pairwise distinct points $t_1,\ldots,t_m\in K$ (not necessarily
ordered), the matrix
$$M\begin{pmatrix}
    u_1,\ldots,u_n\\
    t_1,\ldots,t_m
   \end{pmatrix}
:= (u_k(t_j))_{j=1,\ldots,m;k=1,\ldots,n}
$$
denotes the collocation matrix of $(u_1,\ldots,u_n)$ at the points $t_1,\ldots,t_m$.

\begin{definition}[TP and ASTP sequence, \cite{CarPen:1993}]
The sequence of functions $(u_1,\ldots,u_n)$ on
$K\subset\R$ is called  totally positive (TP) 
 (resp. almost strict totally positive (ASTP)), if all  collocation matrices
$M\begin{pmatrix}
    u_1,\ldots,u_n\\
    t_1,\ldots,t_n
   \end{pmatrix}$ with ordered points $t_1<\ldots<t_n$ in $K$ are TP (resp. ASTP).
\end{definition}

An example of an ASTP sequence are Tschebycheffian B-splines, which we describe in the
next section. The notion of ASTP matrices is useful in order to describe interlacing
properties of interpolation nodes and spline knots, see Theorem \ref{det>0}.


\section{Tschebycheffian B-splines and PEB-splines}

Next we provide some background from \cite{Schum:1981} and \cite{Ron:1987} 
on a class of Tschebycheffian B-splines which
will be used as window functions for Gabor frames in the rest of the article.
We start from positive weight functions $w_j\in C^{m-j}[a,b]$ on an interval
$[a,b]\subset\R$. Then the functions
\begin{align*}
u_1(x)&=w_1(x)\\
u_2(x)&=w_1(x)\int_a^xw_2(s_2)ds_2\\
\vdots\\
u_m(x)&=w_1(x)\int_a^xw_2(s_2)\int_a^{s_2}\cdots\int_a^{s_{m-1}}w_m(s_m)ds_m\cdots ds_2
\end{align*}
are functions in $C^{m-1}[a,b]$, 
which form an \textit{extended complete Tschebycheff} (ECT) system on $[a,b]$.
This means that
$$
\det\left(M\begin{pmatrix}
    u_1,\ldots,u_s\\
    t_1,\ldots,t_s
   \end{pmatrix}\right)>0
$$
for all 
$1\le s\le m$ and $t_1\leq\ldots\leq t_s \in [a,b]$. Here the matrix 
$$M\begin{pmatrix}
    u_1,\ldots,u_s\\
    t_1,\ldots,t_s
   \end{pmatrix}$$ 
is the collocation matrix of Hermite interpolation, if some nodes 
	$t_j$ coincide. The vector space 
	$$\mathcal U_m:={\rm span}(u_1,\ldots,u_m)$$
is called ECT-space. 
	
An important result in spline theory is the existence of Tschebycheffian B-splines (TB-splines) 
$B_m^k\in C^{m-2}(\R)$ of order $m$,
associated with knots $a\le y_k<\ldots<y_{k+m}\le b$, such that 
\begin{align*}
 & {\rm supp}\, B_m^k = [y_k,y_{k+m}],\\
 & B_m^k\mid_{(y_j,y_{j+1})} \in \mathcal U_m,\quad k\le j\le k+m-1.
\end{align*}
Up to normalization, $B_m^k$ is uniquely determined by these properties.
A detailed description of TB-splines of order $m$ is given in \cite[Ch. 9]{Schum:1981}. 
For our purpose, we define the differential operators
\begin{equation}\label{eq:diffop}
{D}_j f = \frac{d}{dx}\left( \frac{f}{w_j} \right),\ 
\ L_j = {D}_j\cdots {D}_1,\quad j=1,\ldots,m,
\end{equation}
as in \cite[page 365]{Schum:1981}. Then the following results are true.

\begin{lemma}
\begin{itemize}
\item[(i) ]$L_j u_k=0$ for $1\le k\le j\le m$, and $\mathcal U_m$ is the kernel of $L_m$.
\item[(ii) ]$(L_j u_{j+1},\ldots,L_j u_m)$ is the ECT-system on $[a,b]$ with weight functions $w_{j+1},\ldots,w_m$, 
called the $j$-th reduced ECT-system.
\item[(iii) ]Let $B_{m-1}^k$ denote the TB-spline of order $m-1$ with knots
$y_k,\ldots,y_{k+m-1}$ and with respect to the first reduced ECT-system. 
Then for every $k\in\Z$ there are constants $a_k,b_k>0$ such that
\begin{equation}\label{eq:Bsplineder}
    L_1 B_m^k = a_k B_{m-1}^k- b_k B_{m-1}^{k+1}.
	\end{equation}
\end{itemize}
\end{lemma}

The assertions (i) and (ii) of the lemma are given in \cite[page 365f]{Schum:1981},
assertion (iii) follows from Theorems 9.23 and 9.28 in \cite{Schum:1981}. 

An important result of \cite[Theorem 9.34]{Schum:1981} 
characterizes the uniqueness of spline interpolation. It requires an interlacing property
of interpolation nodes and spline knots, in a similar way as the Schoenberg-Whitney
conditions in Theorem \ref{SchoenWhit}. Moreover, it covers the case where an 
ordered subsequence of TB-splines is selected, 
with natural ordering of the support intervals
on the real line.
We present this result for the purpose of
finding
Gabor frames $\mathcal{G}(B,\alpha,\beta)$ in the next section, 
where we let $B$ be a special TB-spline in a shift-invariant spline space.

\begin{theorem}[\cite{Schum:1981}]\label{det>0}
Let $\{B_m^k\mid k=1,\ldots,N\}$ be the TB-splines of order $m$ with respect to the knot sequence 
$y_1<\ldots< y_{m+N}$ and weight functions $w_1,\ldots,w_m$. Then for any selection
$1\leq k_1<\cdots<k_s\leq N$ and any $t_1<\cdots<t_s$,
$$
\det\left(M\begin{pmatrix}
    B_m^{k_1},\ldots,B_m^{k_s}\\
    t_1,\ldots,t_s
   \end{pmatrix}\right)\geq 0,
$$
and strict positivity holds if and only if
$$
t_j\in(y_{k_j},y_{k_j+m}),\quad j=1,\ldots,s.
$$
\end{theorem}

In other words,  an ordered sequence of TB-splines constitutes an ASTP sequence, and therefore their collocation 
matrices with ordered
points are ASTP matrices.

Another property of TB-splines is closely related with total positivity
 of the collocation matrices.
 It is called the 
{\em variation-diminishing property}  of TB-splines. It is 
important for applications of TB-splines in Computer Aided Design,
especially in the digital design of curves. 
Following \cite{deBoor:1976}, we say that a function $f:[a,b]\to\R$ 
has at least $p$ strong sign changes, if there exists a nondecreasing sequence
$(\tau_j)_{0\le j\le p}$ in $[a,b]$ with $f(\tau_0)\ne 0$ and, in case $p\ge 1$, 
$f(\tau_{j-1})f(\tau_j)<0$ for all $j=1,\ldots,p$. 
The supremum of the number of strong sign changes of $f$  is denoted by $S^-(f)$.
Similarly, we define the total number of sign changes $S^-(c)$ 
of a sequence of real numbers $c=(c_k)_{0\le k\le N}$.

\begin{theorem}[Variation-diminishing property, see {\cite[Theorem 9.35]{Schum:1981}}]
\label{vardim}
Let $\{B_m^k\mid k=1,\ldots,N\}$ be the TB-splines of order $m$ with respect to the knot sequence 
$y_1<\ldots< y_{m+N}$ and weight functions $w_1,\ldots,w_m$.  
Then $f=\sum_{k=0}^N c_k B_m^k$,
with domain $[a,b]=[y_1,y_{m+N}]$, satisfies 
\begin{equation}\label{eq:vardimTB}
   S^- (f) \le S^-(c).
\end{equation}
\end{theorem}

In order to approach our main concern of studying Zak transforms 
and Gabor frames whose window function is a TB-spline,
we impose the special structure 
of shift-invariance with respect to the integer lattice on the sequence of TB-splines.
This structure is obtained by
\begin{itemize}
\item[(S$_1$) ]letting the knots be integers $y_k=k$, $k\in\Z$, and 
\item[(S$_2$) ] requiring that the weight functions
$w_j$, $1\le j\le m$, have the form
\begin{align*}
   w_j(x) = e^{\alpha_j x}r_j(x),\qquad \alpha_j\in\R,\quad r_j(x+1)=r_j(x) ~\textrm{for all}~x\in\R.
\end{align*}
\end{itemize}
Due to this form of the weight functions,  $B_m^k$ is called a \textit{periodic exponential B-spline} (PEB-spline)
in \cite{Ron:1987},
and if $r_j\equiv 1$ for all $j$,  then $B_m^k$ is called \textit{exponential B-spline} (EB-spline). 
As a consequence of (S$_1$) and (S$_2$), the PEB-splines of order $m$ satisfy 
\begin{equation}\label{eq:shiftinv}
   B_m^k(x)= e^{\alpha_1 k} B_m^0(x-k),\qquad k\in\Z.
\end{equation}
Hence, the space of all spline functions $\sum_{k\in\Z}c_k B_m^k$ with complex coefficients $c_k$
is shift-invariant. We let $B_m:=B_m^0$, for short, and $B_{m-1}:=B_{m-1}^0$ the
PEB-spline of order $m-1$ with respect to the first reduced ECT-system and knots $0,1,\ldots,m-1$.
 Then 
the recursion formula \eqref{eq:Bsplineder} is specified as
\begin{equation}\label{eq:Bsplineder2}
   L_1 B_m = a_{m-1}^{-1} (B_{m-1}-B_{m-1}(\cdot -1)),\qquad a_{m-1}>0,
\end{equation}
cf. \cite[Proposition 3.2]{Ron:1987}. 
Moreover, the identity \eqref{eq:shiftinv} allows us to 
write the variation-diminishing property in \eqref{eq:vardimTB} as
\begin{equation}\label{eq:vardimPEB}
   S^- \biggl( \sum_{k=0}^N c_k B_m(\cdot-k)\biggr) \le S^-(c),
\end{equation}
where the coefficients $c_k$ are real and the sum has 
the domain $[a,b]=[0,m+N]$.

\begin{example}\label{exampleB}  EB-splines of order
 $m\in\N$ (see \cite[page 16f]{Ron:1987}).

Let $\Lambda=(\lambda_1,\ldots,\lambda_m)\in\R^m$, $\lambda_0=0$, and let exponential weight functions
$$w_j = e^{(\lambda_j-\lambda_{j-1})x},\ \ j=1,\ldots,m,$$
be given. Then, with proper normalization, the EB-spline $B_\Lambda:=B_m$, with knots $0,1,\ldots,m$,
  is given by the convolution of the functions $e^{\lambda_j(\cdot)}\chi_{[0,1)}$,
\begin{align}\label{expBspline}
B_\Lambda =  e^{\lambda_1(\cdot)}\chi_{[0,1)} \, \ast e^{\lambda_2(\cdot)}\chi_{[0,1)} \,
 \ast \ldots \ast \, e^{\lambda_m(\cdot)}\chi_{[0,1)}.
\end{align}
Its Fourier transform is
\begin{align}\label{eq:Blambdahat}
\hat{B_\Lambda}(\omega) = \prod_{j=1}^m \frac{e^{\lambda_{j}-2\pi i\omega}-1}{\lambda_{j}-2\pi i\omega}
\end{align}
and the corresponding differential operators in \eqref{eq:diffop} have the form
$$
   L_j =e^{-\lambda_{j}x}\prod_{k=1}^j\left(\frac{d}{dx}-\lambda_{k}\,\id\right).
$$
Note that the EB-spline $B_{m-1}$ of the first reduced ECT-system with knots $0,\ldots,m-1$ is,
up to normalization, given by
\begin{align*}
B_{\{\lambda_2-\lambda_1,\ldots,\lambda_m-\lambda_1\}} =  e^{(\lambda_2-\lambda_1)(\cdot)}\chi_{[0,1)} \,
 \ast \ldots \ast \, e^{(\lambda_m-\lambda_1)(\cdot)}\chi_{[0,1)}.
\end{align*}

For our numerical computations, we use the explicit representation of $B_\Lambda$ 
in each interval $[k-1,k)$, $1\le k\le m$. If 
the numbers $\lambda_1,\ldots,\lambda_m$ are pairwise distinct, then 
the ECT-space is 
$$\mathcal U_m = \mathrm{span}\{ e^{\lambda_j x}\mid 1\le j\le m \}.$$
 For  $\lambda_1<\cdots<\lambda_m$ and $m\geq2$, Christensen and Massopust \cite{ChrisMass:2010} give the closed 
form  
\begin{equation}\notag
  B_\Lambda(x+k-1)=\sum_{j=1}^m \alpha_j^{(k)} e^{\lambda_j x},\qquad x\in[0,1),~1\leq k\leq m,
	\end{equation}
with coefficients 
\begin{align*}
\alpha_j^{(k)} = \begin{cases}
\prod\limits_{r=1,\atop r\neq j}^{m}(\lambda_m-\lambda_r)^{-1},\ \ &k=1\\
(-1)^{k-1} \ \frac{\sum\limits_{1\leq j_1<\cdots<j_{k-1}\leq m,\atop j_1,\cdots,j_{k-1}\neq j}e^{\lambda_{j_1}+\ldots+\lambda_{j_{k-1}}}}{\prod\limits_{r=1,\atop r\neq j}^{m} (\lambda_m-\lambda_r)},\ \ &k=2,\ldots,m.
\end{cases}
\end{align*}
For the general case
$$
   \Lambda=(\underbrace{\xi_1,\ldots,\xi_1}_{s_1},\ldots,
	\underbrace{\xi_r,\ldots,\xi_r}_{s_r}),
$$
with  pairwise distinct $\xi_1,\ldots,\xi_r\in\R$, and each $\xi_j$ 
repeated with multiplicity $s_j\in\N$, we have
$$\mathcal U_m = 
\mathrm{span}\{ e^{\xi_1 x},xe^{\xi_1 x},\ldots,x^{s_1-1}e^{\xi_1 x},\ldots,e^{\xi_r x},\ldots,x^{s_r-1}e^{\xi_r x} \}$$
and
\begin{equation}\label{eq:Bpiecewise}
  B_\Lambda(x+k-1)=\sum_{j=1}^r p_j^{(k)}(x) e^{\xi_j x},\qquad x\in[0,1),~1\leq k\leq m,
\end{equation}
with real polynomials $p_j^{(k)}$ of degree $s_j-1$. 
\end{example}


\section{Zak transform of PEB-splines and TP functions of finite type}

The Zak transform is an important tool for spline interpolation and 
Gabor frame analysis. For a given parameter $\alpha>0$,
 the Zak transform $\mathrm Z_{\alpha}f$ of a function $f:\R\to\C$ is defined by
$$\mathrm Z_{\alpha}f(x,\omega) = \sum_{k\in\Z} f(x-k\alpha) e^{2\pi ik\alpha\omega},$$
whenever this series converges.
Properties of the Zak transform and 
many more facts about this transform can be found in \cite{Groech:2001}.
We present some of these properties which are needed for the results in the following sections.

\begin{lemma}\label{Zakprop}
Let $f$ be an element of the Wiener space $ W(\R)$.
\begin{itemize}
 \item[a) ] $\mathrm Z_\alpha f(x,\omega)$ is bounded in $\R^2$, 
and if $f$ is continuous, then $\mathrm Z_\alpha f$ is continuous.
  \item[b) ] For every $n\in\Z$, we have the identities for periodicity
$$\mathrm Z_{\alpha}f(x,\omega+\tfrac{n}{\alpha}) = \mathrm Z_{\alpha}f(x,\omega)$$
and quasi-periodicity
$$\mathrm Z_{\alpha}f(x+n\alpha,\omega) = 
e^{2\pi in\alpha \omega}\, \mathrm Z_{\alpha}f(x,\omega).$$
 \item[c) ] If $\hat{f}\in W(\R)$ as well, then
$$\alpha\cdot \mathrm Z_{\alpha}f(x,\omega) = e^{2\pi ix\omega}\,\mathrm Z_{1/\alpha}\hat{f}(\omega,-x).$$
 \item[d) ] Let $f_{\alpha}=f(\alpha\cdot)$ be the scaled function of $f$. Then
$$\mathrm Z_{\alpha}f(x,\omega) = \mathrm Z_1f_{\alpha}(\tfrac{x}{\alpha},\alpha\omega).$$
\end{itemize}
\end{lemma}

We study the Zak transform of  PEB-splines $B_m$, with integer knots $0,1,\ldots,m$ and 
with respect to periodic exponential weights $w_j$, $j=1,\ldots,m$ on $[0,\infty)$. 
For the special case $m=1$ (the only case where $B_m$ is
discontinuous), the PEB-spline is
$$
    B_1(x)=w_1(x)\chi_{[0,1)}(x)=e^{\alpha_1 x}r_1(x)\chi_{[0,1)}(x).
$$
Its Zak transform $\mathrm Z_1 B_1$ is 
$$
   \mathrm Z_1B_1(x,\omega)= B_1(x-\lfloor x\rfloor)e^{2\pi i\omega \lfloor x\rfloor},
	\qquad (x,\omega)\in\R^2.
$$
Since $B_1$ is positive on $[0,1)$, 
the Zak transform $\mathrm Z_1 B_1$  has no zeros in $\R^2$. 

For $m\ge 2$, the PEB-spline $B_m$ is continuous and belongs to the Wiener space $W(\R)$.
By Lemma \ref{Zakprop} its Zak transform $\mathrm Z_1B_m$ is continuous.
It was first pointed out in \cite{Zak:1975} that, if the Zak transform $\mathrm Z_1 g$ is continuous, 
then it has a zero in $[0,1)\times[0,1)$. Later this fact was used in \cite{Dau:1990}
to generalize the famous Balian-Low Theorem from orthogonal bases to frames. That Theorem leads to the fact that 
$\mathcal G(g,\alpha,\tfrac{1}{\alpha})$ cannot be a frame, if $g$ is a continuous function in $W(\R)$. 
Since $B_m$ is real, we know from the results in \cite{Janssen:1988}
that there exists a zero 
$\mathrm Z_1B_m(\tilde{x},\tfrac{1}{2}) = 0$ for some $\tilde x\in[0,1)$. 

The absence of more zeros of the Zak transform $\mathrm Z_1g$
was  proved in the following special cases: 
\begin{itemize}
\item   $g=N_m$, the cardinal polynomial B-spline, which is the EB-spline $B_\Lambda$ 
in
Example \ref{exampleB} with $\Lambda=(0,\ldots,0)$: 

A classical result in the theory of cardinal polynomial B-splines,
dating back to  Schoenberg's work in 1973 on cardinal spline interpolation
\cite{Schoenberg:1973}, states 
 that the only zero of $Z_1N_m$ in $[0,1)\times[0,1)$ is located at
$(\tfrac{1}{2},\tfrac{1}{2})$ for even $m$ and $(0,\tfrac{1}{2})$ for odd $m\geq3$. Note that, in Schoenberg's
terminology, the Zak transform $Z_1N_m$ is called {\em exponential Euler spline}. See also
\cite{JetRiemSiv:1991} for a discussion of zeros of these functions.
\item  $g$ is even and \textit{super convex} on $[0,\infty)$: 

Janssen proved in \cite{Janssen:2003.2} that the Zak transform of an even, continuous function $g$ of the form 
$$g(t)=b(t)+b(t+1),\ \ t\geq0,$$ 
where $b$ is integrable, non-negative and strictly convex on $[0,\infty)$, 
has only one zero
in $[0,1)\times[0,1)$, located at $(\tfrac{1}{2},\tfrac{1}{2})$.
\end{itemize}

PEB-splines, in general, do not satisfy any of these assumptions.
Our next result proves that this property of having only one zero in
$[0,1)\times[0,1)$ persists.

\begin{theorem}\label{zakzeroesB}
Let $m\ge 2$ and $B_m$ a 
PEB-spline of order $m$. 
Then $\mathrm Z_1B_m$ has exactly one zero in $[0,1)^2$.
More precisely, there exists  $\tilde{x}\in[0,1)$, such that $\mathrm Z_1B_m(\tilde{x},\tfrac{1}{2}) = 0$,
and $\mathrm Z_1B_m(x,\omega) \neq 0$ 
for all $(x,\omega) \in [0,1)^2\setminus\{(\tilde{x},\tfrac{1}{2})\}$.
\end{theorem}

We will give two independent proofs of Theorem~\ref{zakzeroesB}. The first proof 
makes use of the variation-diminishing property of shifts $B_m(\cdot-k)$ of 
a PEB-spline $B_m$, which was described in Theorem \ref{vardim}. The second proof is more
elementary, but restricted to EB-splines of order $m$. 

\begin{proof} First proof of Theorem~\ref{zakzeroesB}. 

First, we let $\omega\in(-\tfrac{1}{2},\tfrac{1}{2})$ and assume that there exists $\tilde x\in[0,1)$ with \newline
$\zak_1B_m(\tilde x,\omega)=0$. 
By quasi-periodicity of $\zak_1B_m$, the function
\begin{align} \label{eq:ReZakPEB}
  f(x):={\rm Re}\, ( \zak_1B_m( x,\omega))&= \sum_{k\in\Z} c_k B_m(x-k),\quad c_k=\cos(2\pi  k\omega),
\end{align}
vanishes at all points $\tilde x+k$, $k\in\Z$. These points are isolated zeros of $f$
by the following argument. By \cite{Schum:1976}, the TB-splines $B_m^k$ are 
locally linearly independent. This means that, 
if $f$ vanishes on an interval $(a,b)\subset \R$, then all
coefficients $c_k$ with ${\rm supp}\,B_m(\cdot-k) \cap (a,b)\ne\emptyset$ vanish as well.
Since $|\omega|<\tfrac{1}{2}$, no 
consecutive coefficients $c_k$, $ c_{k+1}$ of the function $f$ in \eqref{eq:ReZakPEB} vanish simultaneously. Therefore, due to ${\rm supp}\,B_m=[0,m]$ and $m\ge 2$,
there is no 
non-empty  interval $(\alpha,\beta)$ where $f$ is identically zero.  
Moreover, since the components $f\mid_{[k,k+1]}$ are elements of an ECT-space of dimension $m$,
there are at most $m-1$ isolated zeros in each of these intervals. 

Next, we will find a contradiction to the variation-diminishing property in 
\eqref{eq:vardimPEB}
by counting strong sign changes. As some of the zeros of $f$ may be double zeros, counting the strong sign changes of $f$ 
may not be enough. Instead, we take 
\begin{align*}
   g(x)&=L_1f(x)=\frac{d}{dx}\left(\frac{f}{w_1}\right)= 
	a_{m-1}^{-1}\sum_{k\in\Z} (c_k-c_{k+1}) B_{m-1}(x-k),
\end{align*}
where we used (\ref{eq:Bsplineder2}) and the 
PEB-spline $B_{m-1}$ of order $m-1$. For every $N\in\N$, there are at least
$N+1$ zeros of $f/w_1$ in $[0,N+1]$, and, by Rolle's theorem, we obtain 
$$S^-(g)\ge N\qquad\hbox{on}\quad [0,N+1].
$$
 On the other hand, $g$ is the sum of PEB-splines $B_{m-1}(\cdot-k)$ with real coefficients
\begin{align*}
   d_k=a_{m-1}^{-1}(c_k-c_{k+1})=
	a_{m-1}^{-1}{\rm Re}\left( (1-e^{2\pi i\omega})e^{2\pi i k\omega}\right),\qquad k\in\Z. 
\end{align*}
We notice that
$$|{\rm arg}\left((1-e^{2\pi i\omega})\,e^{2\pi ik\omega}\right) - {\rm arg}(1-e^{2\pi i\omega})|= 
2\pi |k\omega|\quad\hbox{for all}\quad k\in\Z. $$
Therefore, the finite sequence $d=(d_k)_{-m+2\leq k\leq N}$
satisfies
\begin{align}\label{uneq:S(d)}
S^-(d) \le 2(N+m-2)\abs{\omega}. 
\end{align}
Therefore, for $|\omega|<\tfrac{1}{2}$ and large $N$ we find
$$    S^-(d)<N\le S^-(g),$$ 
which is a contradiction to
the variation-diminishing property of $B_{m-1}$ in \eqref{eq:vardimPEB}. This implies, 
that there is no zero of $\zak_1B_\Lambda(\cdot,\omega)$ with $|\omega|<\tfrac{1}{2}$. 

For the case $\omega=\tfrac{1}{2}$, the assumption of having two distinct zeros (or a double zero)
of $\zak_1B_\Lambda(\cdot,\tfrac{1}{2})$ in $[0,1)$ leads to a contradiction by analogous arguments.
\end{proof}

The second proof of Theorem~\ref{zakzeroesB} is more elementary,
but it is restricted to EB-splines $B_\Lambda$ in Example \ref{exampleB}. 
We employ the notation $$\Lambda=(\xi_1,\ldots,\xi_1,\ldots,\xi_r,\ldots,\xi_r)$$
with multiplicities $s_j$ of $\xi_j$ in $\Lambda$ as in Example \ref{exampleB}.

\begin{proof}
For fixed $\omega\in\R$, we
consider the complex function $h:=\zak_1B_\Lambda( \cdot,\omega)$. 
By (\ref{eq:Bpiecewise}), we obtain for $x\in[0,1)$ 
\begin{align*}
  h(x)=\sum_{k=0}^{m-1} B_\Lambda(x+k)e^{-2\pi ik\omega} &=
	\sum_{k=1}^{m} \sum_{j=1}^r p_j^{(k)}(x)e^{\xi_jx} e^{-2\pi ik\omega} \\
       &= \sum_{j=1}^r \underbrace{\sum_{k=1}^{m} p_j^{(k)}(x)e^{-2\pi ik\omega}}_{:=q_j(x)} 
			e^{\xi_jx},
\end{align*}
where $q_j$ are complex polynomials of degree $\sigma_j-1\le s_j-1$. 
 Some of the $q_j$ are nonzero, as pointed out
in the first proof, by local linear independence of the EB-splines. 
We let $\sigma_j=0$ if $q_j=0$
and define $\mu=\sigma_1+\ldots+\sigma_r$. 
Without loss of generality, we can assume $\sigma_1\ge 1$, that is, the term
$q_1(x)e^{\xi_1 x}$  in $h|_{[0,1)}$ is nonzero. 
We use the identity
$$
    e^{\xi_j x} \frac{d}{dx}\left( e^{-\xi_j x}h(x)\right)=q_j'(x)e^{\xi_j x}+\sum_{k\ne j} 
		((\xi_k-\xi_j)q_k(x)+q_k'(x)))e^{\xi_k x}.
$$ 
Writing $D_j$ for the differential operator on the left hand side and \newline
$\mathcal{D}:=D_{1}^{\sigma_1-1}\prod_{j=2}^{r}D_{j}^{\sigma_j}$, we obtain
$$
  \mathcal{D} h(x)
= be^{\xi_1x},\qquad x\in (0,1),
$$
with  a nonzero constant $b\in\C$. The quasi-periodicity of $h$ leads directly to
\begin{align*}
   \mathcal{D}h(x)= be^{\xi_1(x-k)}\,e^{2\pi ik\omega},\qquad x\in (k,k+1),
\end{align*}
for all $k\in\Z$. 

Now we let $\omega\in(-\tfrac{1}{2},\tfrac{1}{2})$ and we assume 
that there exists $\tilde x\in[0,1)$ with
$\zak_1B_\Lambda(\tilde x,\omega)=0$. 
By quasi-periodicity of $h=\zak_1B_\Lambda(\cdot,\omega)$, the function
$ f={\rm Re}\, h$
vanishes at all points $\tilde x+k$, $k\in\Z$, and these points are isolated zeros of $f$
by the same argument of local linear independence as before.
This guarantees that $f$ has at least $N\in\N$ isolated zeros in $[0,N]$. 
Note that $\mathcal{D}$ is a differential operator of order $\mu-1\le m-1$. 
Since  $f\in C^{m-2}(\R)$, with $f^{(m-2)}$ absolutely continuous,
we obtain by Rolle's
theorem that
\begin{equation}\label{eq:signchange}
  S^-(\mathcal{D}f)\ge N-\mu+1\qquad \hbox{on}\quad [0,N].
	\end{equation}
However, 
on each interval $[k,k+1)$ with $k\in\Z$, the sign of $\mathcal{D}f$ is fixed by
$$
   {\rm sign}\,(\mathcal{D}f)(x) = 
	{\rm sign}\, {\rm Re}\left(b\,e^{2\pi ik\omega}\right),\qquad x\in[k,k+1).
$$
In the same way as in (\ref{uneq:S(d)}), this implies
$$
   S^-(\mathcal{D}f)\le  2N|\omega|\qquad\hbox{on}\quad [0,N].$$
This is a contradiction to \eqref{eq:signchange} for $|\omega|<\tfrac{1}{2}$ and 
large $N$.

An analogous argument leads to a contradiction, if we assume that $\omega=\tfrac{1}{2}$ and there
are two distinct zeros of $Z_1B_\Lambda$ with $x\in[0,1)$. 
\end{proof}

\begin{remark}\label{zakzeroesC}
The result of Theorem~\ref{zakzeroesB} can be slightly generalized.
Let $m\ge 2$ and $B_m$ a PEB-spline of order $m$.
Consider the function 
$g=\sum_{l=0}^r a_l B_m(\cdot-l)$, where the coefficients $a_l$ define
the trigonometric polynomial \newline 
$\hat a(\omega)=\sum_{l=0}^r a_le^{-2\pi il\omega} $ with no real zeros.
Then the Zak transform $\zak_1g$ is
\begin{align*} 
\zak_1g(x,\omega)&= \sum_{k\in\Z} \sum_{l=0}^r a_l B_m(x-k-l)e^{2\pi i k\omega}\\ 
&= \hat a(\omega) \zak_1B_m(x,\omega).
\end{align*}
Hence, the zeros of $\zak_1 g$ are the same as the zeros of $\zak_1 B_m$, and 
the result of Theorem~\ref{zakzeroesB} extends to the Zak transform $\zak_1 g$. 
\end{remark}

Next we turn our attention to TP functions $g$ of finite type $m\in\N$, given by the
Fourier transform
$$\hat{g}(\omega) = \prod_{\nu=1}^m (1+2\pi i \omega{a_{\nu}}^{-1})^{-1},$$
where $a_1,\ldots,a_m\in\R\setminus\{0\}$. 
We are able to express the Zak transform $\zak_\alpha g$,
with arbitrary $\alpha>0$, in terms of the Zak transform $\zak_1 B_\Lambda$ of an EB-spline of order $m$.
This has several advantages:
\begin{itemize}
\item[1.] Since an EB-spline $B_\Lambda$ of order $m$ has support $[0,m]$, its Zak transform 
can be computed accurately for given values of $x$ and $\omega$, without the need of truncation.
\item[2.] The property that $\mathcal{G}(g,\alpha,\beta)$ is a Gabor frame for some pair $(\alpha,\beta)$
of lattice parameters is equivalent to $\mathcal{G}(B_\Lambda,1,\alpha\beta)$ being a Gabor frame.
\item[3.] Finding frame-bounds of $\mathcal{G}(g,\alpha,\beta)$ 
is reduced to finding frame-bounds for $\mathcal{G}(B_\Lambda,1,\alpha\beta)$.
\end{itemize}
More details on items 2. and 3. will be presented in the last two sections of this article.
We mention that 
another explicit representation 
of $\zak_\alpha g$, without the use of EB-splines, 
was recently found in \cite{BanGroeStoe:2013}.

\begin{theorem}\label{b-splinezak}
Let $\alpha>0$ and $g\in L^1(\R)$ be a TP function of finite type, defined by its Fourier transform
$$\hat{g}(\omega) = \prod_{\nu=1}^m  (1+2\pi i \omega{a_{\nu}}^{-1})^{-1},$$
where $a_1,\ldots,a_m\in\R\setminus\{0\}$. 
With $\lambda_{\nu}:=-\alpha a_{\nu}$ and $B_\Lambda$ defined in (\ref{expBspline}), we have
$$\alpha\,\mathrm Z_{\alpha}g(x,\omega) = \prod_{\nu=1}^m \, \frac{\alpha a_{\nu}}{1-e^{-\alpha(a_{\nu}+2\pi i\omega)}}\, 
\mathrm Z_1B_\Lambda(\tfrac{x}{\alpha},\alpha\omega),\ \ (x,\omega)\in[0,\alpha)\times[0,\tfrac{1}{\alpha}).$$
\end{theorem}

\begin{proof}
The Fourier transform of $B_\Lambda$ in \eqref{eq:Blambdahat} is
\begin{align*}
\hat{B}_\Lambda(\omega) &=
  \prod_{\nu=1}^m \, \frac{e^{\lambda_{\nu}-2\pi i\omega}-1}{\lambda_{\nu}
	(1-2\pi i\omega\lambda_{\nu}^{-1})}\\
  &= \underbrace{\prod_{\nu=1}^m \, \frac{1 - e^{-(\alpha a_{\nu}+2\pi i \omega)}}{\alpha a_{\nu}}}_{=:t(\omega)} \
	\prod_{\nu=1}^m(1+2\pi i\omega(\alpha a_{\nu})^{-1})^{-1} \\
	&=t(\omega)	\ \alpha\,\hat{g}_{\alpha}(\omega),
\end{align*}
where $t$ is a one-periodic function in $\omega$ and $g_{\alpha} = g(\alpha\cdot)$. This implies
\begin{align*}
\mathrm Z_1\hat{B}_\Lambda(\omega,x) = t(\omega)\, \alpha\, \mathrm Z_1\hat{g}_{\alpha}(\omega,x)
\end{align*}
and finally, since $g,B_\Lambda\in W(\R)$, with Lemma~\ref{Zakprop} we get
\begin{align*}
\alpha\,\mathrm Z_{\alpha}g(x,\omega) &= \alpha\,\mathrm Z_1g_{\alpha}(\tfrac{x}{\alpha},\alpha\omega) \\
&= \alpha\,\mathrm Z_1\hat{g}_{\alpha}(\alpha\omega,-\tfrac{x}{\alpha})\, e^{2\pi ix\omega} \\
&= t(\alpha\omega)^{-1}\,\mathrm Z_1\hat{B}_\Lambda(\alpha\omega,-\tfrac{x}{\alpha})\, e^{2\pi ix\omega} \\
&= t(\alpha\omega)^{-1}\, \mathrm Z_1B_\Lambda(\tfrac{x}{\alpha},\alpha\omega),
\end{align*}
which completes the proof.
\end{proof}

Clearly, the factor $t(\alpha\omega)$ in the proof is nonzero and bounded. Hence, we 
obtain the following corollary of Theorem \ref{zakzeroesB}.

\begin{corollary}\label{zakzeroesg}
Let  $g$ be a totally positive function of finite type $m\ge 2$ and $\alpha >0$. 
Then there exists $\tilde{x}\in[0,\alpha)$, such that 
$\mathrm Z_{\alpha}g(\tilde{x},\tfrac{1}{2\alpha}) = 0$, and $\mathrm Z_{\alpha}g(x,\omega) \neq 0$ for all 
$(x,\omega) \in [0,\alpha)\times[0,\tfrac{1}{\alpha})\setminus\{(\tilde{x},\tfrac{1}{2\alpha})\}$.
\end{corollary}

With reference to our paper, 
the result of Corollary \ref{zakzeroesg}
 was already used in \cite{BanGroeStoe:2013} for the construction of 
periodic and discrete Gabor frames $\mathcal{G}(g,\alpha,\tfrac{1}{\alpha})$ 
for the spaces $L^2([0,K])$ 
and $\mathbb C^K$, respectively, at the critical density. In other words, families of
Riesz bases for $L^2([0,K])$ 
or bases for $\mathbb C^K$ are constructed in \cite{BanGroeStoe:2013}, 
which result from periodization resp. discretization
of Gabor families with a TP window function of finite type at the critical density.


\section{Gabor frames of PEB-splines and TP functions of finite type}

In this section, we prove that every PEB-spline $B_m$ of order $m$  is
the window function of a Gabor frame $\mathcal{G}(B_m,\alpha,\beta)$, for a large set
of lattice parameters $\alpha,\beta>0$. 
In combination with Theorem \ref{b-splinezak},
 this will
provide an alternative proof of the result in \cite{GroeStoe:2012}, that
every Gabor family $\gabor$ with TP function $g$ of finite type
and with lattice parameters $\alpha,\beta>0$,
$\alpha\beta<1$, is a frame for $L^2(\R)$, see Theorem \ref{TPframe} below.

\begin{theorem}\label{PEBsplineFrame}
Let $m\in\N$ and $B_m$ be a PEB-spline with knots $0,\ldots,m$. 
Then the set $\mathcal G(B_m,\alpha,\beta)$ constitutes a frame in the following cases:
\begin{enumerate}
\item[(1) ]$0<\alpha<m$ and $0<\beta\le m^{-1}$,
\item[(2) ]$\alpha\in \{1,2,\ldots,m-1\}$, $\beta>0$ and $\alpha\beta<1$,
\item[(3) ]$\alpha>0$, $\beta\in \{1,2^{-1},\ldots,(m-1)^{-1}\}$ and  $\alpha\beta<1$.
\end{enumerate}
\end{theorem}

\begin{proof}
The upper frame-bound
of the Gabor family $\mathcal{G}(B_m,\alpha,\beta)$ is finite, because $B_m$ is piecewise 
continuous with compact support, so it is an element of the 
Wiener space $W(\R)$. Finding a positive lower frame-bound is the more difficult part.

For case  (1), we have $\beta^{-1}\ge{\rm vol}({\rm supp}\, B_m)$.
Hence, the assertion is a direct consequence of 
the more general result in  \cite{DauGroMey:1986}, 
which states that the optimal lower frame-bound is 
given by
\begin{equation}\label{eq:Daubbound}
   A_{\rm opt}=\frac{1}{\beta}  \essinf_{x\in[0,\alpha)} \sum_{k\in\Z} B_m(x+k\alpha)^2>0.	
\end{equation}

Let us now consider the cases (2) and (3), so 
$\alpha$ or $\beta^{-1}$ are in $\{1,2,\ldots,m-1\}$. This excludes the case
$m=1$, so we only consider $m\ge 2$ from now on. In addition, by the result for case (1),
we can assume that
$$0<\alpha<\beta^{-1} < m.$$

The main idea of the proof that $\mathcal{G}(B_m,\alpha,\beta)$ has a lower frame-bound
is taken from \cite[Theorem 8]{GroeStoe:2012}. The proof is constructive and produces a 
dual window function $\gamma$ such that $\mathcal{G}(\gamma,\alpha,\beta)$
is a dual frame for $\mathcal{G}(B_m,\alpha,\beta)$.
The construction is based on  \cite[Lemma 5]{GroeStoe:2012}, which we repeat here
in the univariate setting for the reader's convenience.

\begin{lemma}[\cite{GroeStoe:2012}]\label{zerorow-lemma}
Let $g\in W(\R)$ and $\alpha,\beta>0$. 
Assume that there exists a Lebesgue measurable vector-valued function
$\sigma (x)$ from $\R$ to $\ell ^2(\Z )$ with period $\alpha $, such that
\begin{equation}\label{eq:c6}
  \sum _{k\in \Z } \sigma _k(x) \bar{g} (x+\alpha k -
  \tfrac{l}{\beta})  = \delta _{l,0} \qquad \text{a.a. } x \in \R  \, .
\end{equation}
If 
\begin{equation}\label{eq:c7}
\sum _{k\in \Z } \sup _{x\in [0,\alpha ]} |\sigma _k(x)| < \infty ,
\end{equation}
then  $\gabor $ is a frame. Moreover, with
$$\gamma (x) =  \beta  \sum _{k\in \Z } \sigma _{k}(x) \chi _{[0,\alpha
    )}(x-\alpha k),\qquad x\in\R,$$
the set $\mathcal{G}(\gamma,\alpha,\beta)$   is a dual frame of $\gabor$.
\end{lemma}

Now we make use of this lemma in our proof of Theorem~\ref{PEBsplineFrame}.
We define the pre-Gramian matrix
\begin{equation}\notag
P(x)=\left( B_m(x+ k\alpha- \tfrac{l}{\beta}) \right)_{k,l\in\Z},\qquad x\in\R.
\end{equation}
Let 
$$
  k_0:=\left\lceil \frac{m-\tfrac{1}{\beta}}{\tfrac{1}{\beta}-\alpha}\right\rceil-1.
$$
We select the compact interval 
$$
   I=	  [m-\tfrac{1}{\beta},m-\tfrac{1}{\beta}+\alpha]
$$
 of length $\alpha$ inside the support of $B_m$, 
and for every $x\in I$, 
we construct the row vector $\sigma(x)$ such that
\eqref{eq:c6} holds, in other words 
\begin{equation}\label{eq:sigmaP}
   \sigma(x)  P(x)=(\delta_{l,0})_{l\in\Z}.
\end{equation}
For every $x\in I$, the inequalities
\begin{align*}  
&  0<m-\tfrac{1}{\beta}\le x \le m-\left(\tfrac{1}{\beta}-\alpha\right)<m,\\
&  0< x-k_0\left(\tfrac{1}{\beta}-\alpha\right)\le \tfrac{1}{\beta}<m
\end{align*}
are readily verified. Therefore, the rows $k=0,1,\ldots,k_0$ of $P(x)$ 
can be described as follows:
\begin{itemize}
\item all diagonal entries 
$B_m\left(x-k\left(\tfrac{1}{\beta}-\alpha\right)\right)$, $k=0,1,\ldots,k_0$, 
are strictly positive,
\item all entries in columns $l<0$  vanish, because  
$x+k\alpha-\tfrac{l}{\beta}\ge m$, 
\item all entries in columns $l>k_0$  vanish, 
because $x+k\alpha-\tfrac{l}{\beta}\le 0$.
\end{itemize}
Therefore, the only nonzero block of $P(x)$ in rows $0\le k\le k_0$ is the square matrix
\begin{equation}\notag
P_0(x):=\left(B_m(x+ k\alpha- \tfrac{l}{\beta})\right)_{k,l=0,\ldots,k_0}
\end{equation}
with strictly positive diagonal entries.
For case (3), where $\beta^{-1} \in\{ 1, 2,\ldots,m-1\}$, the matrix $P_0(x)$ is the 
collocation matrix
$$
    P_0(x)= M\begin{pmatrix} B_m & B_m(\cdot-\beta^{-1}) &\cdots & B_m(\cdot-k_0\beta^{-1})\\
		x&x+\alpha & \cdots & x+k_0\alpha\end{pmatrix}.
$$
By Theorem \ref{det>0}, $P_0(x)$ is invertible. Likewise, for case (2), where \newline
$\alpha \in\{ 1, 2,\ldots,m-1\}$, the transpose $P_0(x)^T$ is the 
collocation matrix
$$
    P_0(x)^T= M\begin{pmatrix} B_m & B_m(\cdot+\alpha) &\cdots & B_m(\cdot+k_0\alpha)\\
		x&x-\beta^{-1} & \cdots & x-k_0\beta^{-1}\end{pmatrix},
$$
with reverse ordering of the support of the B-splines and nodes as compared to Theorem \ref{det>0}.
Again we conclude that $P_0(x)$ is invertible. Next we define 
\begin{equation}\label{eq:sigmak0}
   (\sigma_0(x),\ldots,\sigma_{k_0}(x)),\qquad x\in I,
\end{equation}
to be the first row of the inverse of $P_0(x)$, and extend this vector with
 components $\sigma_k(x)=0$ for $k<0$ and $k> k_0$. Then the identity \eqref{eq:sigmaP}
follows immediately. 

It remains to show that $\sigma(x)$ satisfies 
\begin{equation}\label{eq:c7I}
\sum _{k\in \Z } \sup _{x\in I} |\sigma _k(x)| < \infty ,
\end{equation}
which is equivalent to \eqref{eq:c7} by periodicity. Recall that $m\ge 2$, so
$B_m$  is continuous. The determinant of $P_0(x)$ is strictly positive and
 continuous as a function of $x\in I$. Since $I$ is compact, there is a constant $c>0$ such that
$\det P_0(x)\ge c$ for all $x\in I$. By Cramer's rule, there is a uniform bound for
all entries of $P_0(x)^{-1}$, and this gives \eqref{eq:c7I}. 
\end{proof}

\begin{remark}
Gabor frames with window function $B_m$ were also considered in 
\cite{ChrisMass:2010}. In their work, the lattice parameters 
$0<\alpha<m$, $0<\beta\le \frac{1}{2m}$ define highly redundant Gabor frames, and
explicit formulas for  dual windows $\gamma$ 
with support $[0,m]$ are presented. 
\end{remark}

\begin{center}
\begin{figure}
\includegraphics[width=14.0cm]{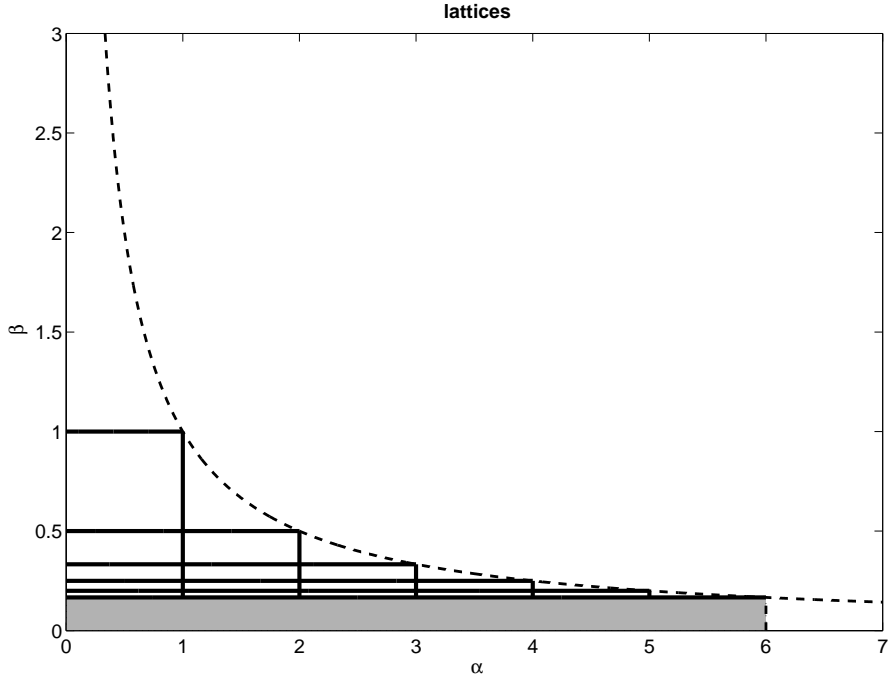}

\caption{Lattice parameters $(\alpha,\beta)$, where every PEB-spline of order $m=6$ induces a Gabor frame, include the 
rectangle $0<\alpha<6$, $0<\beta\le 1/6$ and the lines $\alpha\in\{1,\ldots,5\}$ and 
$\beta^{-1}\in\{1,\ldots,5\}$ with $\alpha\beta<1$.}
\label{PEBspline_lattices}
\end{figure}
\end{center}


\begin{remark}
Different dual windows for the Gabor frame $\mathcal{G}(B_m,\alpha,\beta)$ can be computed, if we choose larger 
rectangular submatrices $P_1(x)$ of $P(x)$ in the following way.
We include more columns of $P(x)$ 
to the left and right of $P_0(x)$ and select the maximal set of rows of $P(x)$, such that 
only zeros appear on both sides of $P_1(x)$. Similar arguments as in the proof
of Theorem \ref{PEBsplineFrame} show that the matrix $P_1(x)$ has 
full column rank. We take its Moore-Penrose pseudo-inverse $\Gamma_1$
and define the central part of the vector $\sigma(x)$ in \eqref{eq:sigmak0}
by the row vector with index $k=0$ of $\Gamma_1$.
The matrix $P_1(x)$ often has a
 better $\ell_2$-condition number than $P_0(x)$,
and the corresponding dual window $\gamma$ has smaller norm  $\|\gamma\|_W$.
For this case, the proof of (\ref{eq:c7}) needs some small 
adaptations as in \cite{Klo:2012}, where
 TP functions were considered. For comparison of this construction of the dual window $\gamma$,
we mention that the canonical dual window $\tilde \gamma$ 
is defined by the row vector $\sigma(x)$ in row $k=0$ of the Moore-Penrose
pseudo-inverse 
$$\Gamma(x)=(P(x)^TP(x))^{-1}P(x)^T$$
of the (full) pre-Gramian $P(x)$, in the same way that
 $\tilde \gamma(x+k\alpha)=\beta \sigma_k(x)$ for all $k\in\Z$.
\end{remark}

\begin{example}
We consider the EB-spline
 $B_{\Lambda_1}$, with
$\Lambda_1 = (-2,-1,1,2)$, drawn in the top left of Figure~\ref{PEBduals}.
Two different dual windows $\gamma$ of the Gabor frame
$\mathcal{G}(B_{\Lambda_1},1,0.86)$ are computed by different selections of the matrix $P_1(x)$ 
with
an increasing number of columns of the pre-Gramian $P(x)$. The dual window functions $\gamma$ 
are drawn in the left half of Figure~\ref{PEBduals}. They are
discontinuous at the boundary points of the interval $I$ of length $1$
chosen for the computation of the row vector $\sigma(x)$ in the proof of
Theorem \ref{PEBsplineFrame}, and the 
integer shifts of these points.
The larger the submatrix $P_1(x)$, 
the smaller is the jump at these discontinuities. We also observe that the dual window $\gamma$ 
approaches the canonical dual window $\tilde \gamma$, as we increase the number of columns of $P_1(x)$. 

Similar computations are performed with 
the EB-spline $B_{\Lambda_2}$, $\Lambda_2 = (1,2,3)$, shown on the top right of Figure~\ref{PEBduals}. 
Two different dual windows $\gamma$ of the Gabor frame
$\mathcal{G}(B_{\Lambda_2},2,0.45)$ 
are drawn in the right half of Figure~\ref{PEBduals}.
\end{example}

\begin{center}
\begin{figure}
\includegraphics[width=14.0cm]{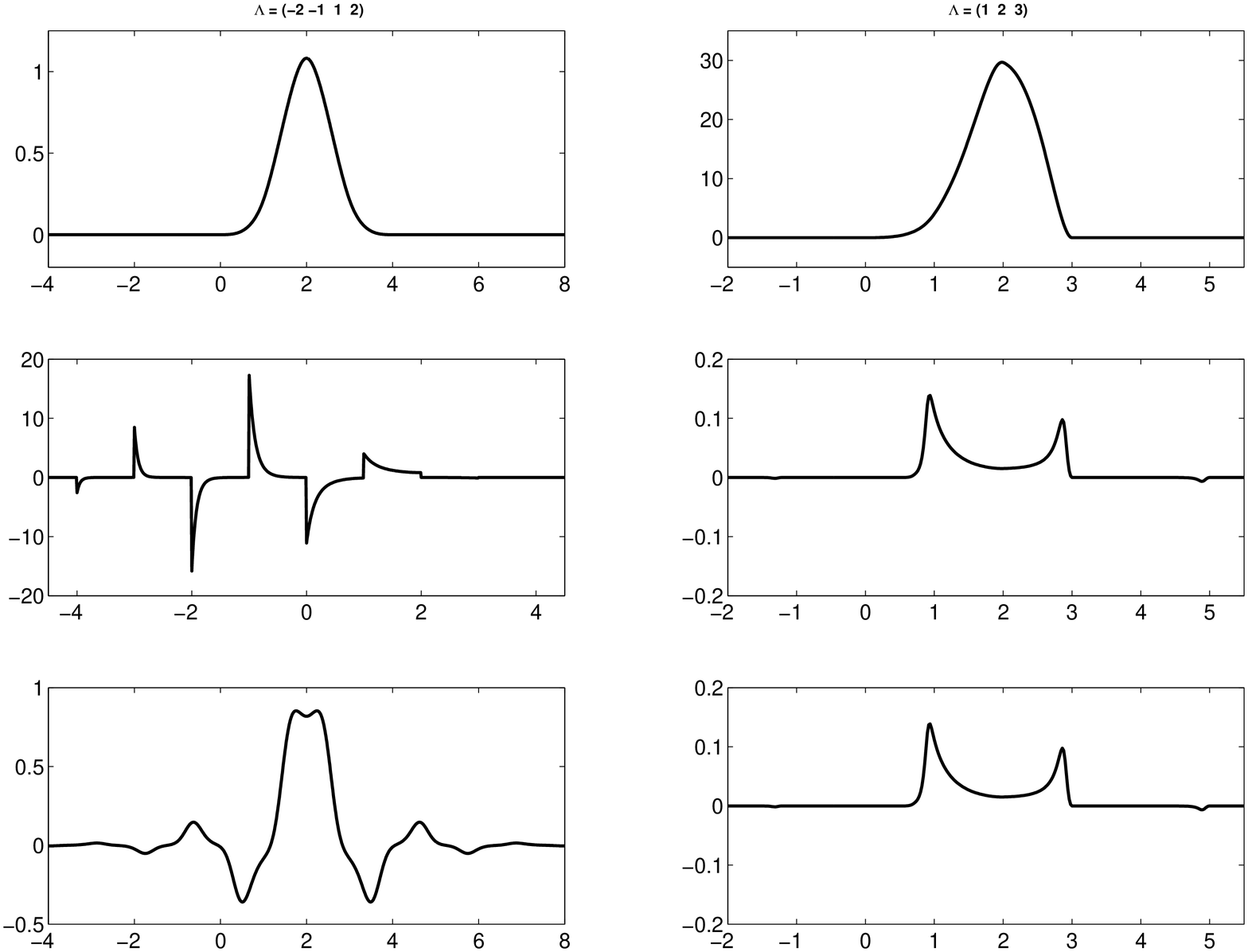}

\caption{EB-spline $B_\Lambda$ with $\Lambda = (-2,-1,1,2)$ 
and two different dual windows for $(\alpha,\beta)=(1,0.86)$ (left side);
EB-spline $B_\Lambda$ with  $\Lambda = (1,2,3)$ 
and two different dual windows for $(\alpha,\beta)=(2,0.45)$ (right side).}
\label{PEBduals}
\end{figure}
\end{center}

The relation between the Zak transform of TP functions and EB-splines
in Theorem \ref{b-splinezak} can be used 
in order to compare frame-bounds of the different Gabor systems.
The general method behind was first described by Janssen and Strohmer 
\cite{JansStrohm:2002} and used for
 comparing the Gaussian and the hyperbolic 
secant as window functions for Gabor frames. 
This method is based on
the following characterization of frame-bounds given by Ron and Shen  \cite{RonShen:1997}. 

\begin{theorem}\label{RonShenChar}
Let $g\in L^2(\R)$ be a 
window-function, $\alpha,\beta> 0$ some lattice-parameters, and
$P(x)=(\overline{g}(x+k\alpha-\tfrac{l}{\beta}))_{k,l\in\Z}$
the pre-Gramian. 
Then $0<A\leq B<\infty$ 
are frame-bounds for $\mathcal G(g,\alpha,\beta)$, if and only if
\begin{align*}
\inf_{x\in[0,\alpha]} &\zweinorm{P(x) \, c}^2 \geq \beta A \, \zweinorm{c}^2,\\
\sup_{x\in[0,\alpha]} &\zweinorm{P(x) \, c}^2 \leq \beta B \, \zweinorm{c}^2,\ \ 
\forall c\in\ell^2(\Z).
\end{align*}
\end{theorem}

If $g\in W(\R)$ and $c\in\ell_1(\Z)$, 
an application of Parseval's identity reveals that
\begin{align*}
\zweinorm{P(x)c}^2 &= \sum_{k\in\Z} \left| \sum_{l\in\Z} c_l \overline{g}(x+\alpha k- \tfrac{l}{\beta}) \right|^2 \\
&= \alpha \int_0^{1/\alpha}  
   \left| \sum_{k,l\in\Z} c_l \overline{g}(x+\alpha k- \tfrac{l}{\beta}) e^{-2\pi i k\alpha \omega}\right|^2 \,d\omega \\
&= \alpha \ \norm{\sum_{l\in\Z} \, c_l \, \mathrm{Z}_{\alpha}\overline{g}(x-\tfrac{l}{\beta},\cdot)}_{L^2(0,\frac{1}{\alpha})}^2.
\end{align*}
Based on this identity, the observation by Janssen and Strohmer in \cite{JansStrohm:2002} can be summarized as follows. 

\begin{proposition}\label{trick}
Let $\alpha,\beta,\kappa>0$. 
If $\gabor$ is a frame with frame-bounds $0<A\leq B<\infty$, and $h\in L^2(\R)$ satisfies
$$  \mathrm Z_{\kappa\alpha}h(\kappa x,\kappa^{-1}\omega) = 
D(\omega)\cdot \mathrm Z_{\alpha}g(x,\omega),\ \ 
(x,\omega)\in[0,\alpha)\times[0,\tfrac{1}{\alpha}),$$
where
$$0< \essinf\abs{D(\omega)}^2 \leq \esssup\abs{D(\omega)}^2 < \infty,$$
then $\mathcal G(h,\kappa\alpha,\kappa^{-1}\beta)$ is also a frame with frame-bounds 
$\kappa A\cdot\essinf\abs{D(\omega)}^2$ 
and $\kappa B\cdot\esssup\abs{D(\omega)}^2$.
\end{proposition}

If $g$ is a TP function of finite type,
we found the factorization of $\zak_\alpha g$ in 
 Theorem~\ref{b-splinezak}. 
 Moreover, together with Theorem~\ref{PEBsplineFrame} this leads to an alternative proof
of \cite[Theorem 1]{GroeStoe:2012}.

\begin{theorem}\label{TPframe}
Let $g$ be a totally positive function of finite type $m\geq2$. Then $\gabor$ is a frame for every $\alpha,\beta>0$ with $\alpha\beta<1$.
\end{theorem}


\section{Lower frame-bounds for Gabor frames of even two-sided exponentials}

The pre-Gramian matrices 
$P(x)=(B_m(x+k\alpha-\tfrac{l}{\beta}))_{k,l\in\Z}$ were 
already used in the proof of Theorem \ref{PEBsplineFrame}. We observed 
that, for $\alpha=1$, the finite submatrices are collocation 
matrices of PEB-splines, hence
they are ASTP matrices.  
In this section, 
we apply methods for the decomposition of ASTP matrices in order to
compute lower frame-bounds for Gabor frames in a special case, namely for the 
symmetric EB-spline of order $2$, given by 
\begin{equation}\label{eq:Blambda}
B(x) = (e^{\lambda(\cdot)}\chi_{[0,1]}\ast e^{-\lambda(\cdot)}\chi_{[0,1]})(x)
=\begin{cases}\frac{\sinh (\lambda x)}{\lambda}\,,&\qquad 0\le x\le 1,\\
\frac{\sinh (\lambda(2- x))}{\lambda}\,,&\qquad 1< x\le 2,
\end{cases}
\end{equation}
with $\lambda\in\R_+$ and ${\rm supp}\,B=[0,2]$. Then 
we use Theorem \ref{b-splinezak} once more for deriving
the corresponding lower frame-bound of the Gabor frame $\gabor$, where 
$g$ is  the even two-sided exponential function
$$
  g(x) = \tfrac{\lambda}{2}e^{-\lambda\abs{x}}.
$$

Note that $g$
is a TP function of finite type $2$.
It was shown by Janssen  \cite{Janssen:2003}
that $\gabor$ is a frame for all $\alpha,\beta >0$ with $\alpha\beta<1$, 
see also Theorem \ref{TPframe}.
It is also known that the lower frame-bound $A$ 
of these frames approaches $0$ when
$\alpha\beta$ tends to $1$. Our goal of this section 
is to find a quantitative 
result which describes these properties.

\begin{remark}\label{lowerframebound-remark}
There are only two window functions (and scaled versions thereof) in the literature, 
where the asymptotic behaviour of the 
lower frame-bound was specified near the critical density $\alpha\beta\approx 1$.
These are the Gabor frames $\mathcal{G}(g,\alpha,\alpha)$ 
with $g$ the Gaussian window $g(x)=e^{-\pi x^2}$ and the hyperbolic secant 
$g(x)=(\cosh \pi x)^{-1}$. It was proved in
\cite{BoGroLy:2010} that constants $c_1,c_2>0$ exist  
such that the optimal lower frame-bound $A_{opt}$ 
of $\mathcal{G}(g,\alpha,\alpha)$
satisfies 
$$
   c_1 (1-\alpha) \le A_{opt} \le c_2(1-\alpha)\quad \hbox{for}\quad
	 \tfrac{1}{2}<\alpha < 1.
$$
In this section, we give explicit lower bounds of the same linear
asymptotic decay as $\alpha\beta$ tends to $1$, 
if the window function is the symmetric EB-spline \eqref{eq:Blambda} of order 2
and $\alpha=1$
(Theorem~\ref{bound}),
or the TP function $g(x)=\tfrac{\lambda}{2}e^{-\lambda|x|}$ of finite type $2$
and general $\alpha,\beta>0$ (Theorem \ref{TPBound}).
\end{remark}

\begin{example} The left-hand side of Figure~\ref{fig:framebound}  depicts
lower frame-bounds 
for the Gabor frame $\mathcal{G}(B,1,\beta)$  with the window function 
$B$ in  \eqref{eq:Blambda} and $\lambda=1$.
The points mark the optimal lower frame-bound $A_{opt}$ for rational
values
$\beta> 1/2$, taken at $\beta=k/61$, $31\le k\le 60$.
The computation of the optimal lower frame-bound
follows the method of Zibulski and Zeevi,
based on a matrix-valued Zak transform of  $B$, see \cite{ZibZee:1997} or \cite[Theorem 8.3.3]{Groech:2001}.
The lower frame-bound $A$ in Theorem~\ref{bound}
is drawn as a solid line. Note that, by the distinction of 
three different cases for $\tfrac{1}{2}<\beta<1$, the lower bound has a jump
at $\beta=\tfrac{5}{6}$. The asymptotically linear decay of $A$ near $\beta=1$ is expreessed
by 
$$
    A^{-1} = \left( 2\sinh^2(\tfrac{1}{2}) (1-\beta)  \right)^{-1} + \mathcal O\left((1-\beta)^{-1/2}\right).
$$
The right half of Figure~\ref{fig:framebound} gives the same type of information 
for the Gabor frame $\mathcal{G}(g,1,\beta)$, where $g(x)=\tfrac{1}{2}e^{-|x|}$.
The solid line
shows the lower frame-bound $A$ in Theorem \ref{TPBound}.
\end{example}

\begin{center}
\begin{figure}
\includegraphics[width=14.0cm]{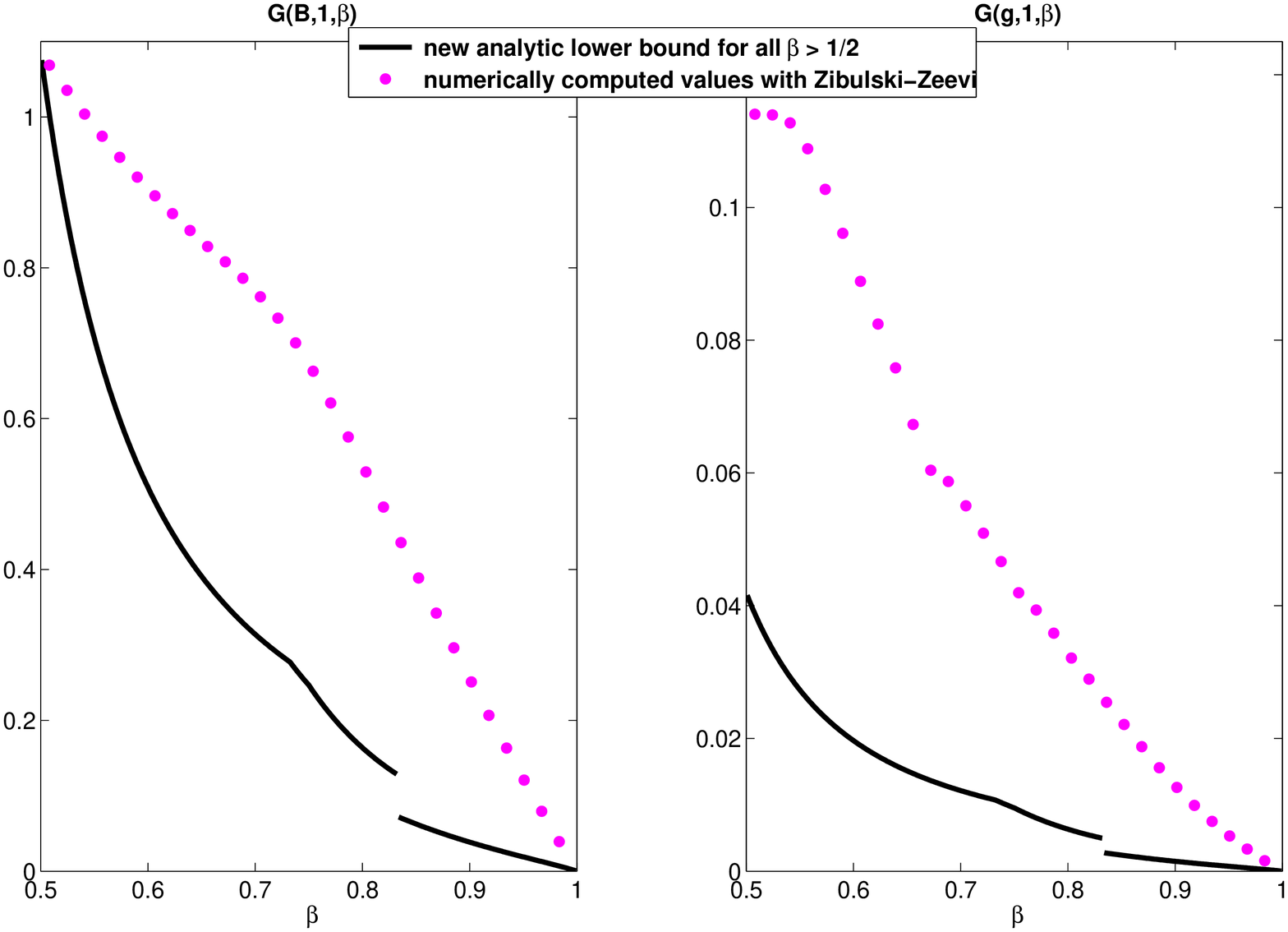}

\caption{Points in the left figure show the optimal lower frame-bound
 of $\mathcal G(B,1,\beta)$ 
for $\beta=k/61$, $31\le k\le 60$, where $B$ is the EB-spline \eqref{eq:Blambda} with 
$\lambda=1$. The solid line in the left figure depicts the lower frame-bound $A$ in 
Theorem \ref{bound}. Points in the right figure show the optimal lower frame-bound
 of $\mathcal G(g,1,\beta)$, where $g(x)=\tfrac{1}{2}e^{-|x|}$.
The solid line in the right figure depicts the lower frame-bound in 
Theorem \ref{TPBound}.}
\label{fig:framebound}
\end{figure}
\end{center}

We start with the discussion of lower frame-bounds for $\mathcal{G}(B,1,\beta)$,
where $B$ is the EB-spline in (\ref{eq:Blambda}).
For this purpose, we inspect the pre-Gramian matrices
$$P(x)=\left( {B}(x+ k- \tfrac{l}{\beta}) \right)_{k,l\in\Z},\qquad x\in\R,$$
in more detail.

Before we present our result in Theorem~\ref{bound},
let us illustrate our method for finding lower frame-bounds for $\mathcal{G}(B,1,\beta)$
by first considering lattices with high redundancy. Let
$1/\beta \geq 2=\mathrm{vol}(\mathrm{supp}(B))$. 
As stated before in \eqref{eq:Daubbound},
the optimal lower frame-bound 
 is given by 
\begin{equation}\notag
A_{opt} = \inf_{x\in[0,1)} \frac{1}{\beta} \sum_{k\in\Z} \abs{B(x+ k)}^2.
\end{equation}
This was, for example, published in \cite{DauGroMey:1986} and is also part of the Casazza-Christensen bound \cite{Chris:2003}. 
We next
 show how we can find this optimal lower frame-bound
$A_{\beta}(B)$ from the pre-Gramian $P(x)$. 
For every $x\in[0,1)$, there is at most one non-zero entry in each row of 
$P(x)$. Namely, if $\beta=1/2$, we have
\begin{equation}\label{eq:PBbeta}
P(x) = 
\begin{pmatrix}
 \ddots &        &                &                & \\
        & B(x)   &                &                & \\
        & B(x+1) &                &                & \\
        &        & B(x) &                & \\
        &        & B(x+1) &                & \\
        &        &                & B(x) & \\
        &        &                & B(x+1) & \\
        &        &                &                & \ddots
\end{pmatrix}, \ \ x\in[0,1].
\end{equation}
For $\beta<1/2$, the entry of $P(x)$ in row $k$ and column $l$ is given by
$$p_{k,l}(x)=B(x+k-\tfrac{l}{\beta}).$$ 
This entry is nonzero, if and only if
$$
   	k\in I_l:= \Z\cap (\tfrac{l}{\beta}-x,\tfrac{l}{\beta}-x+2).
$$
The sets $I_l$ are non-empty and, due to $\beta<1/2$,  they are pairwise disjoint. This implies that the matrix $P(x)$ has a similar form as (\ref{eq:PBbeta}), 
with the possibility of 
some zero rows. The Moore-Penrose pseudo-inverse of $P(x)$ is
$$
   \Gamma(x)
     =(P(x)^T P(x))^{-1} P(x)^T,
$$ 
and $P(x)^T P(x)$ is a diagonal matrix with diagonal entries
$$
   \sum_{j\in\Z} p_{j,l}(x)^2=\sum_{j\in\Z} B(x+j\alpha-\tfrac{l}{\beta})^2,
	\qquad l \in\Z.$$
Hence, 
the entry of $\Gamma(x)$ in row $l$ and column $k$ is
$$
   \gamma_{l,k}(x)= \frac{B(x+k-\tfrac{l}{\beta})}
	{\sum_{j\in\Z} B(x+j\alpha-\tfrac{l}{\beta})^2}.
$$
Since nonzero entries in different rows of $\Gamma(x)$ 
appear in disjoint sets of col\-umns, 
the operator norm is 
$$
   \|\Gamma(x)\|_2 = 
	\sup_{l\in\Z} \left(\sum_{k\in\Z} 
	\gamma_{l,k}(x)^2\right)^{1/2}=
	\sup_{l\in\Z} \left(\sum_{j\in\Z} 
	B(x+j\alpha-\tfrac{l}{\beta})^2\right)^{-1/2}.
$$
Moreover,  $\Gamma(x)$ is a left-inverse of $P(x)$ with 
the smallest $\ell_2$-operator norm. Therefore, with Theorem~\ref{RonShenChar} 
 the optimal lower frame-bound  of the Gabor frame 
$\mathcal{G}(B,1,\beta)$ with $0<\beta\le \frac{1}{2}$ is given by
$$
 A_{opt}=\beta^{-1}\left(\sup_{x\in [0,1)}\|\Gamma(x)\|_2\right)^{-2} =
\beta^{-1}\inf_{x\in [0,1)} ( B(x)^2+B(x+1)^2).
$$
This infimum is obtained for $x=\tfrac{1}{2}$ and has the  value
\begin{equation}\label{eq:AoptB}
 A_{opt}= \frac{2\sinh^2(\tfrac{\lambda}{2})}{\beta \lambda^2}.
\end{equation}

Our next theorem extends this construction to frames $\mathcal{G}(B,1,\beta)$
with  $\beta>1/2$. It is no longer true that consecutive columns of 
$P(x)$ are supported in disjoint sets of rows. However, 
we can still find a block decomposition
of the pre-Gramian $P(x)$ and a suitable left-inverse in order
to prove the following result.

\begin{theorem}\label{bound}
Let $\lambda>0$, $0<\beta<1$, and  $B$ be the EB-spline in (\ref{eq:Blambda}).
Then the corresponding Gabor frame 
$\mathcal G(B,1,\beta)$ has the  lower frame-bound
\begin{align}\label{eq:boundex}
A=  ( \beta\lambda^2c_\beta)^{-1} 
\min\left\{2\sinh^2(\tfrac{\lambda}{2}),
~\sinh^2(\tfrac{\lambda}{2\beta})\right\},
\end{align}
where 
\begin{align*}
  c_\beta = \begin{cases}
	\ds  1\ ,  &\mathrm{if\ }0<\beta\leq\tfrac{1}{2},\\
  \ds (3-\beta^{-1})^{2}\ ,  &\mathrm{if\ }\tfrac{1}{2}<\beta\leq\tfrac{3}{4},\\
  \ds (3-\beta^{-1})(11-11\beta^{-1}+3\beta^{-2})\ ,  
	&\mathrm{if\ }\tfrac{3}{4}<\beta\leq\tfrac{5}{6},\\
  \ds     \left(1+\sqrt{\frac{\beta}{\pi(1-\beta)}}\right)
	\left(1+\sqrt{\frac{\pi\beta}{4(1-\beta)}}\right),	&\mathrm{if\ }\tfrac{5}{6}<\beta<1.
\end{cases}
\end{align*}
\end{theorem}

\begin{remark}
The parameter $\lambda=0$ can be included. Then $B$ in \eqref{eq:Blambda} 
is the  
linear polynomial B-spline 
$$B(t) = (\chi_{[0,1]}\ast \chi_{[0,1]})(t)
=\begin{cases}t\,,&\qquad 0\le t\le 1,\\
(2- t)\,,&\qquad 1< t\le 2,
\end{cases}
$$
and the lower frame-bound of $\mathcal{G}(B,1,\beta)$ is
$$
   A=(\beta c_\beta)^{-1}\cdot \min\left\{1/2,
~1/(2\beta)^2)\right\} .
$$
\end{remark}

\begin{proof}

For $\beta \le \tfrac{1}{2}$ the result was shown in \eqref{eq:AoptB}.

Now let $1/2<\beta<1$ and $x\in[0,1)$. 
As in our motivation,  
we  construct a left-inverse $\Gamma(x)$ of the (full) pre-Gramian 
$$P(x)=\left(p_{k,l}(x)\right)_{k,l\in\Z}=
\left(B(x+k-\tfrac{l}{\beta})\right)_{k,l\in\Z}$$ 
with operator norm
$$
   \|\Gamma(x)\|_2	\le  c_\beta^{1/2}\, \lambda	
	 \left(\min\left\{\sqrt{2}\sinh(\tfrac{\lambda}{2}),
~\sinh(\tfrac{\lambda}{2\beta})\right\}\right)^{-1}.
$$
Then the result follows from Theorem~\ref{RonShenChar}.

First, we
take a close look at the form of $P(x)$ and detect a block structure 
which is  more complicated than for $\beta=1/2$ in (\ref{eq:PBbeta}). 
Since the support of $B$ is $[0,2]$, there are at most two non-zero entries in each row and in 
each column of $P(x)$.  Moreover, every $2\times 2$ submatrix of $P(x)$ contains at least one zero.
The overall form of $P(x)$ is a staircase, similar to (\ref{eq:PBbeta}), with 
steps of height one if 
$$
    x+k-\tfrac{l}{\beta}, ~x+k+1-\tfrac{l+1}{\beta}  \in (0,1],
$$
and steps of height two if 
$$
    x+k-\tfrac{l}{\beta}  \in (0,1],\quad  ~x+k+1-\tfrac{l+1}{\beta}<0.
$$
This form is depicted in the following scheme, where a step of height
two appears between the first and second column, followed by 
two steps of height one and so on,
\begin{align*}
P(x) = 
\begin{pmatrix}
 \ddots & * & 0 & 0 & 0 & 0 & 0 & \\
        & * & 0 & 0 & 0 & 0 & 0 & \\
        & 0 & * & 0 & 0 & 0 & 0 & \\
        & 0 & * & * & 0 & 0 & 0 & \\
        & 0 & 0 & * & * & 0 & 0 & \\
        & 0 & 0 & 0 & * & 0 & 0 & \\
        & 0 & 0 & 0 & 0 & * & 0 & \\
        & 0 & 0 & 0 & 0 & * & * & \\
        & 0 & 0 & 0 & 0 & 0 & * & \ddots
\end{pmatrix}.
\end{align*}
We see that the pre-Gramian  has the following block-structure.
\begin{itemize}
\item[1.] If, for some $k,l\in\Z$, we have $x+k-\tfrac{l}{\beta}=1$, then 
the $(k,l)$-entry $p_{k,l}(x)=B(1)=\tfrac{\sinh\lambda}{\lambda}$ defines a  $1$-by-$1$ block, 
with no other non-zero entries in row $k$ and column $l$ of $P(x)$. 
\item[2.] All other blocks have the form 
\begin{equation}\label{eq:P0proof}
P_0:=
\begin{pmatrix}
   a_1 &  &  &    \\
         b_1 & a_2 &  &    \\
          & \ddots & \ddots &    \\
          &  & b_{s-1} & a_{s}  \\
          &  &  & b_s   
\end{pmatrix}  \in \R^{(s+1)\times s}
\end{equation}
with one more row than column and only zeros to the right and left, top and bottom of this block.
The first entry $a_1=B(x_1)$ is the entry in row $k$ and column $l$ of $P(x)$ 
with 
\begin{equation}\notag
    x_1:=x+k-\tfrac{l}{\beta}\in (2-\tfrac{1}{\beta},1),
\end{equation}
such that only zeros appear to the left and right of $a_1$ in row $k$ of $P(x)$ and also 
above in column $l$. 
Subsequent entries are
\begin{equation}\label{eq:aj}
   a_j=B(x_j),\quad b_j=B(x_j+1)=B(1-x_j),\qquad 1\le j\le s,
\end{equation}
where 
\begin{equation}\label{eq:tj}
x_j:=x_1-(j-1)(\tfrac{1}{\beta}-1),\qquad 1\le j\le s,
\end{equation}
and $s\in\N$ is defined by the condition that
\begin{equation}\notag
   x_s=x_1-(s-1)(\tfrac{1}{\beta}-1)\in (0,\tfrac{1}{\beta}-1].
\end{equation}
Note that only zeros appear to the left and right of $b_s$ in the corresponding row of $P(x)$
and also below in the corresponding column.  
Simple computations reveal that $s=\lceil \beta x_1/(1-\beta)\rceil$. 
So blocks have   $s\approx \lceil \tfrac{\beta}{1-\beta}\rceil$ columns.
\end{itemize}   

A left-inverse $\Gamma(x)$ of $P(x)$ is constructed by considering each block separately.
The operator norm of $\Gamma(x)$, as an operator on $\ell_2(\Z)$,
 is the supremum of the operator norms of all blocks.

Now, let us consider a typical block $P_0$ of $P(x)$. 
\begin{itemize}
\item[1.] If $P_0$ is $1$-by-$1$, then its inverse is $\Gamma_0=\tfrac{\lambda}{\sinh\lambda}$. 
\item[2.] If $P_0$ is $s+1$-by-$s$, with $s$ as described above, then we use a decomposition of $P_0$ 
which is similar to the Neville-elimination method in \cite{GasPen:1995}. 
The following lemma provides an upper bound of the norm of a left-inverse $\Gamma_0$
and thus completes the proof of Theorem~\ref{bound}.
\end{itemize}
\end{proof}

\begin{lemma}\label{PBneville-lemma}
Let $\tfrac{1}{2}<\beta<1$ and 
 $P_0$ be the matrix in (\ref{eq:P0proof}), with $a_j,b_j$  in (\ref{eq:aj}),
 $1\le j\le s$. 
Then there exists 
a left-inverse $\Gamma_0$ of $P_0$ such that 
$$
   \|\Gamma_0\|_2	\le  c_\beta^{1/2}\, \lambda	
	 \left(\min\left\{\sqrt{2}\sinh(\tfrac{\lambda}{2}),
~\sinh(\tfrac{\lambda}{2\beta})\right\}\right)^{-1}
$$
with the constant $c_\beta$ in Theorem~\ref{bound}. 
\end{lemma}

\begin{proof}
Note that $P_0$ is an ASTP matrix and has full rank.
Indeed, the $s$-by-$s$ minor of the first $s$ rows is positive
as its diagonal entries are positive, see
Theorem~\ref{det>0}. 
We  construct a decomposition of $P_0$ which is similar to the 
Neville elimination,
as explained in \cite{GasPen:1995},
with slight modifications due to the rectangular form of $P_0$. 

For this purpose, let 
$$
   \delta =\tfrac{1}{2}(\tfrac{1}{\beta}-1),\qquad \omega=\frac{1}{4\delta\beta}=\frac{1}{2-2\beta}.
$$	
Note that $0<\delta<\tfrac{1}{2}$ and $\omega>1$. We start from the observation that 
$$
   1>x_1\ge 2-\frac{1}{\beta} = 1-2\delta>\frac{1}{2}-\delta>0,\qquad 
	0<x_s\le 2\delta <\frac{1}{2}+\delta.
$$
Therefore, we can choose $1\le r\le s$ such that
\begin{align}\notag
   |x_r-\tfrac{1}{2}|\le \delta. 
\end{align}
This is always possible, as the sequence of points $x_j$ in 
(\ref{eq:tj})
 decreases with uniform stepsize $2\delta$, and we have
\begin{equation}\label{eq:tjsequence}
   1>x_1>\cdots>x_{r-1}\ge \tfrac{1}{2}+\delta\ge x_r\ge
	\tfrac{1}{2}-\delta\ge x_{r+1}>\cdots>x_s>0,
\end{equation}
with the possibility of empty subsequences in these inequalities,  
if $r=1$ or $r=s$. Due to $1>x_1\ge\tfrac{1}{2}+(2r-3)\delta$,  we 
find by simple computations that $r-1< \omega$.

Now we construct a left-inverse of $P_0$. 
In contrast to Gauss elimination, the 
Neville elimination uses neighboring rows.
In this manner, we obtain the factorization
$$
  P_0= CS
$$
where
\begin{align*}
C=&\begin{pmatrix}
   1 &  &  &    \\
         \tfrac{b_1}{a_1} & 1 &  &    \\
          & \ddots & \ddots &    \\
          &  & \tfrac{b_{r-1}}{a_{r-1}}& 1  \\
          &  & & & 1 & \tfrac{a_{r+1}}{b_{r+1}}&   \\
          &  & & &   & \ddots & \ddots &    \\
          &  & & &   &        & 1 &\tfrac{a_{s}}{b_{s}}\\
					&  & & &   &        &   &  1	\\
\end{pmatrix} ,\\
S=&\begin{pmatrix}
   a_1 &  &  &    \\
          & a_2 &  &    \\
          & & \ddots &    \\
          &  &  & a_{r}  \\
          &  &  & b_r  \\
          &  &  & & & \ddots \\
					&  &  & & & & b_{s-1}\\
          &  &  & & & & & b_s   
\end{pmatrix}  .
\end{align*}
The Moore-Penrose pseudo-inverse $S^\dagger$ of $S$ is readily computed and gives
$$
   \|S^\dagger\|_2 = \max\left\{ a_j^{-1}\textrm{ for } 1\le j\le r-1,~~b_j^{-1}\textrm{ for } r+1\le j\le s,~~
	(a_r^2+b_r^2)^{-1/2}\right\}.
$$
By (\ref{eq:aj}) and (\ref{eq:tjsequence}) we obtain 
$$
    \|S^\dagger\|_2\le \max\left\{ \frac{1}{B(\tfrac{1}{2}+\delta)}, 
		\frac{1}{\sqrt{2} B(\tfrac{1}{2})}\right\}=
		 \max\left\{ \frac{\lambda}{\sinh(\lambda(\tfrac{1}{2}+\delta))}, 
		\frac{\lambda}{\sqrt{2} \sinh(\tfrac{\lambda}{2})}\right\}.
$$
Next we find an upper bound for the norm of $F=C^{-1}$. Elementary linear algebra
shows that $F$ has the block form
$$
  F=\begin{pmatrix}
	   L &\\  & U
	\end{pmatrix}
$$
with lower triangular matrix  $L$ and upper triangular matrix $U$, and entries
\begin{align*}
   f_{j,k}&= (-1)^{j-k} \prod_{l=k}^{j-1} \frac{b_l}{a_l},&\qquad 1\le k\le j\le r,\\
   f_{s+2-j,s+2-k}&= (-1)^{j-k} \prod_{l=k}^{j-1} \frac{a_{s+1-l}}{b_{s+1-l}},&\qquad 1\le k\le j\le s-r+1.\\
\end{align*}
Clearly, $\|F\|_2=\max\{\|L\|_2,\|U\|_2\}$, by the block structure of $F$. 
We use Schur's test for finding an upper bound of $\|L\|_2$, and only mention 
here that the estimates for $\|U\|_2$ 
follow analogously.  

First, we obtain
from (\ref{eq:tj}) and $x_{r-1}\ge \tfrac{1}{2}+\delta$,  that 
$$
    1>x_{r-j}\ge  \frac{1}{2}+(2j-1)\delta, \qquad 1\le j\le r-1.
$$
Simple calculations give
$$
  \frac{\sinh\lambda x}{\sinh\lambda(1-x)}\le 2x\quad\textrm{for all}\quad x\in[0,\tfrac{1}{2}],
$$
and thus
\begin{align*}
   \frac{b_{r-j}}{a_{r-j}}&= \frac{\sinh(\lambda(1-x_{r-j}))}{\sinh (\lambda x_{r-j})}\le 
	2-2x_{r-j} 
	\le 1-(4j-2)\delta 
	=1-\frac{j-\tfrac{1}{2}}{\omega-\tfrac{1}{2}}.
\end{align*}
By inserting these upper bounds,
and by use of the basic inequality $\ln(1-x)\le -x$ for $x\in [0,1]$,
we obtain 
\begin{align*}
  |f_{j,k}|&\le \prod_{l=k}^{j-1} 
	\left(1-\frac{r-l-\tfrac{1}{2}}{\omega-\tfrac{1}{2}}\right)\\
	&\le \exp\left(-\sum_{l=k}^{j-1}\frac{r-l-\tfrac{1}{2}}{\omega-\tfrac{1}{2}}\right)\\
	&	=\exp\left(-\frac{(j-k)(2r-j-k)}{2\omega-1}\right)
\end{align*}
for all $	1\le k\le j\le r$.
Note that for $j=k$ we have the empty product $f_{j,j}=1$. 

Let us turn to the estimate of $\|L\|_\infty$. 
The sum of absolute values of all entries in row $1\le j\le r$ of $L$ is bounded by
\begin{align*}
    \rho_j = \sum_{k=1}^{j} |f_{j,k}| &\le 
		\sum_{k=1}^j \exp\left(-\frac{(j-k)(2r-j-k)}{2\omega-1}\right)\\
		&=\sum_{k=0}^{j-1} \exp\left(-\frac{k(2r-2j+k)}{2\omega-1}\right).
\end{align*}
The last sum is maximal for $j=r$,
and the interpretation as a Riemann sum gives
\begin{align*}
  \|L\|_\infty &\le \sum_{k=0}^{r-1} 
	\exp\left(-\frac{k^2}{2\omega-1}\right) \\
	&\le
	1+\int_0^\infty \exp\left(-\frac{x^2}{2\omega-1}\right)\,dx
	=1+\sqrt{\frac{(2\omega-1)\pi}{4}} .
\end{align*}
By the identity $2\omega-1=\tfrac{\beta}{1-\beta}$, we obtain the upper bound
$$
   \|L\|_\infty \le 1+\sqrt{\frac{\pi\beta}{4(1-\beta)}}.
$$

The same method provides a slightly smaller bound for
the sums of absolute values in column $k$ of $L$.
More precisely, for $1\le k\le r$, we have 
\begin{align*}
    \kappa_k = \sum_{j=k}^{r} |f_{j,k}| &\le 
		\sum_{j=k}^r 
	\exp\left(-\frac{(j-k)(2r-j-k)}{2\omega-1}\right) \\
	&=\sum_{j=0}^{r-k} 
	\exp\left(-\frac{j(2r-2k-j)}{2\omega-1}\right) \\
	&\le
	1+\int_0^{r-k} \exp\left(-\frac{x(2r-2k-x)}{2\omega-1}\right)\,dx\\
	&=1+\sqrt{2\omega-1}\int_0^c e^{-x(2c-x)}\,dx ,
\end{align*}
where we let $c:= (r-k)(2\omega-1)^{-1/2}$ in the last step.
The function $h(c)=\int_0^c e^{-x(2c-x)}\,dx$ achieves its maximum over all positive values of $c$
for $c_0\approx 0.9241$, and $h(c_0)=1/(2c_0)\approx 0.5410< \pi^{-1/2}$. Therefore, we obtain
$$
   \|L\|_1 \le 1+\sqrt{\frac{\beta}{\pi(1-\beta)}}.
$$
Schur's test gives 
$$
   \|L\|_2^2 \le \|L\|_1\|L\|_\infty \le \left(1+\sqrt{\frac{\beta}{\pi(1-\beta)}}\right)
	\left(1+\sqrt{\frac{\pi\beta}{4(1-\beta)}}\right).
$$
More accurate bounds of $\|L\|_1$ and $\|L\|_\infty$ can be given for small values of 
$r$, e.g. for $r=2$ (this is the case $\tfrac{1}{2}<\beta\le \tfrac{3}{4}$), we have  
\begin{align*}
   &L=\begin{pmatrix} 1&0\\ f_{2,1}&1\end{pmatrix},
	\qquad |f_{2,1}|\le 1-\frac{\tfrac{1}{2}}{\omega-\tfrac{1}{2}}
	=2-\frac{1}{\beta},\\
	&\|L\|_1=\|L\|_\infty \le 3-\frac{1}{\beta}.
\end{align*}
Likewise, for $r=3$ (this is the case $\tfrac{3}{4}<\beta\le \tfrac{5}{6}$), we have
\begin{align*}
   &L=\begin{pmatrix} 1&0&0\\ f_{2,1}&1&0\\ f_{3,1}& f_{3,2}&1\end{pmatrix},\\
	&|f_{2,1}|\le 1-\frac{\tfrac{3}{2}}{\omega-\tfrac{1}{2}}=4-\frac{3}{\beta},
	\quad |f_{3,2}|\le 1-\frac{\tfrac{1}{2}}{\omega-\tfrac{1}{2}}
	=2-\frac{1}{\beta},\quad
	|f_{3,1}|=|f_{2,1} f_{3,2}|,\\
	&\|L\|_1\le 3-\frac{1}{\beta},\qquad 
	\|L\|_\infty \le 11-\frac{11}{\beta}+\frac{3}{\beta^2}.
\end{align*}
These bounds lead to the particular upper bounds in Lemma~\ref{PBneville-lemma}.
Therefore, the proof of the lemma is complete.
\end{proof}

%

We can now add a similar result for Gabor frames $\gabor$, where $g$ is 
the two-sided exponential 
$$
    g(x)=\tfrac{\lambda}{2}e^{-\lambda |x|}
$$ 
with $\lambda>0$.
The relation in Theorem \ref{b-splinezak} gives 
$$
   \zak_\alpha g(x,\omega) =
	\frac{1}{\alpha} C(\omega) \zak_1 B(\tfrac{x}{\alpha},\alpha\omega),
$$
where $B$ is the symmetric EB-spline of order $2$ 
in \eqref{eq:Blambda} with parameter $\alpha\lambda$
instead of $\lambda$, and
$$
   C(\omega)=\frac{\alpha\lambda}{1-e^{-(\alpha\lambda+2\pi i\omega)}}~
	\frac{-\alpha\lambda}{1-e^{-(-\alpha\lambda+2\pi i\omega)}}.
$$
The minimum of $|C(\omega)|$ occurs at $\omega=\tfrac{1}{2}$,
$$
   |C(\omega)|\ge |C(\tfrac{1}{2})|=
	\frac{\alpha^2\lambda^2}{4\cosh^2(\tfrac{\alpha\lambda}{2})}.
$$
With Proposition \ref{trick} and Theorem \ref{bound}, 
where we replace $\lambda$ by $\alpha\lambda$ and $\beta$ by $\alpha\beta$ 
in \eqref{eq:boundex}, we obtain the following lower frame-bound.

\begin{theorem}\label{TPBound}
Let $g$ be the two-sided exponential with $g(x)=\tfrac{\lambda}{2}e^{-\lambda|x|}$, $\lambda>0$, and $\alpha,\beta>0$ some lattice parameters with $\alpha\beta<1$.
Then the  Gabor frame $\gabor$ has the  lower frame-bound
\begin{align*}
A=  
\frac{ \lambda^2\ \min\left\{2\sinh^2(\tfrac{\alpha\lambda}{2}),
~\sinh^2(\lambda/(2\beta))\right\}}
{16\beta c_{\alpha\beta}\ \cosh^4(\tfrac{\alpha\lambda}{2})} ,
\end{align*}
where $c_{\alpha\beta}$ is defined as in Theorem~\ref{bound}.
\end{theorem}


\section*{Acknowledgement}

\thispagestyle{empty}
We are grateful to J. M. Pe\~na and T. Springer for  helpful discussions concerning  the 
estimates of the lower frame-bounds in the last section.
\clearpage




\end{document}